%% file: paper.tex
\documentclass[11pt]{article}
\pdfoutput=1

\usepackage{fullpage}
\usepackage{graphicx}   
\usepackage{subfig}  
\usepackage{hyperref}   
\usepackage{mdwlist}
\usepackage{xspace}
\usepackage[usenames,dvipsnames]{color}
\usepackage{amsmath, amsthm,amssymb}   

\input{macros}

\usepackage{mathpazo} 
\usepackage[scaled]{helvet} 
\usepackage{courier} 
\normalfont
\usepackage[T1]{fontenc}

\begin{document}

\title{Skew in Parallel Query Processing}

\author{
Paul Beame, Paraschos Koutris and Dan Suciu\\
       University of Washington, Seattle, WA\\
       \{beame,pkoutris,suciu\}@cs.washington.edu
}

\maketitle

\begin{abstract}
  We study the problem of computing a conjunctive query $q$ in
  parallel, using $p$ of servers, on a large database.  We consider
  algorithms with one round of communication, and study the complexity
  of the communication.  We are especially interested in the case
  where the data is skewed, which is a major challenge for scalable
  parallel query processing.  We establish a tight connection between
  the {\em fractional edge packings} of the query and the {\em amount
    of communication}, in two cases.  First, in the case when the only
  statistics on the database are the cardinalities of the input
  relations, and the data is skew-free, we provide matching upper and
  lower bounds (up to a poly log $p$ factor) expressed in terms of
  fractional edge packings of the query $q$.  Second, in the case when
  the relations are skewed and the heavy hitters and their frequencies
  are known, we provide upper and lower bounds (up to a poly log $p$
  factor) expressed in terms of packings of {\em residual queries}
  obtained by specializing the query to a heavy hitter.  All our lower
  bounds are expressed in the strongest form, as number of bits needed
  to be communicated between processors with unlimited computational
  power.  Our results generalizes some prior results on uniform
  databases (where each relation is a
  matching)~\cite{DBLP:conf/pods/BeameKS13}, and other lower bounds
  for the MapReduce model~\cite{ASSU13}.
\end{abstract}

\input{section1-introduction}

\input{model}

\input{section2-basic}
\input{section3-algorithm}

\input{section4-lower-bound}

\input{mapreduce}

\input{section5-conclusions}

\bibliographystyle{abbrv}
\bibliography{bib}

\input{appendix}

\end{document}

%% file: macros.tex
\usepackage{graphicx}   
\usepackage{hyperref}   
\usepackage{mdwlist}
\usepackage{xspace}
\usepackage[usenames,dvipsnames]{color}
\usepackage{verbatim}   

\usepackage[normalem]{ulem}

\newcommand{\cut}[1]{}   

\newenvironment{packed_item}{
\begin{itemize}
   \setlength{\itemsep}{1pt}
   \setlength{\parskip}{0pt}
   \setlength{\parsep}{0pt}
}
{\end{itemize}}

\newenvironment{packed_enum}{
\begin{enumerate}
   \setlength{\itemsep}{1pt}
  \setlength{\parskip}{0pt}
   \setlength{\parsep}{0pt}
}
{\end{enumerate}}

\newcommand{\ie}{{\em i.e.}\xspace}

\newcommand{\introparagraph}[1]{\textbf{#1.}}  

\newcommand{\set}[1]{\{#1\}}                    
\newcommand{\setof}[2]{\{{#1}\mid{#2}\}}        

\usepackage{aliascnt} 

\newtheorem{theorem}{Theorem}[section]          	
\newaliascnt{lemma}{theorem}				
\newtheorem{lemma}[lemma]{Lemma}              	
\aliascntresetthe{lemma}  					
\newaliascnt{conjecture}{theorem}			
\aliascntresetthe{conjecture}  				
\newaliascnt{remark}{theorem}				
              
\aliascntresetthe{remark}  					
\newaliascnt{fact}{theorem}				
              
\aliascntresetthe{fact}  					
\newaliascnt{corollary}{theorem}			
\newtheorem{corollary}[corollary]{Corollary}      
\aliascntresetthe{corollary}  				
\newaliascnt{definition}{theorem}			
\newtheorem{definition}[definition]{Definition}    
\aliascntresetthe{definition}  				
\newaliascnt{proposition}{theorem}			
\newtheorem{proposition}[proposition]{Proposition}  
\aliascntresetthe{proposition}  				
\newaliascnt{example}{theorem}			
\newtheorem{example}[example]{Example}  	
\aliascntresetthe{example}  				

\providecommand{\iprod}[1]{{\langle#1\rangle}}
\providecommand{\norm}[1]{{\lVert#1\rVert}}

\providecommand{\bx}[0]{\mathbf{x}}
\providecommand{\bm}[0]{\mathbf{m}}
\providecommand{\bM}[0]{\mathbf{M}}
\providecommand{\bh}[0]{\mathbf{h}}

\providecommand{\bu}[0]{\mathbf{u}}

\providecommand{\ba}[0]{\mathbf{a}}
\providecommand{\be}[0]{\mathbf{e}}
\providecommand{\E}[0]{\mathbf{E}}

\providecommand{\mA}[0]{\mathcal{A}}

\newcommand{\mpc}[0]{MPC}

%% file: section1-introduction.tex
\section{Introduction}

\label{sec:intro}

While in traditional query processing the main complexity is dominated
by the disk access time, in modern massively distributed systems the
dominant cost is that of the communication.  A data analyst will use a
cluster with sufficiently many servers to ensure that the entire data
fits in main memory.  Unlike MapReduce~\cite{DBLP:conf/osdi/DeanG04}
which stores data on disk between the Map and the Reduce phase for
recovery purposes, newer systems like Spark~\cite{spark:2012} and its
SQL-extension Shark~\cite{DBLP:conf/sigmod/XinRZFSS13} perform the
entire computation in main memory, and use replay to recover.  In
these systems the new complexity parameter is the communication cost,
both the amount of data sent and the number of rounds.

A key requirement in such systems is that the data be uniformly
partitioned on all servers, and this requirement is challenging to
enforce when the input data is {\em skewed}.  A value in the database
is skewed, and is called a {\em heavy hitter}, when its frequency is
much higher than some predefined threshold.  Since data reshuffling is
typically done using hash-partitioning, all records containing a heavy
hitter will be sent to the same server, causing it to be overloaded.
Skew for parallel joins has been studied intensively since the early
days of parallel databases, see~\cite{DBLP:conf/vldb/WaltonDJ91}.  The
standard join algorithm that handles skew is the {\em skew
  join}~\cite{DBLP:conf/sigmod/OlstonRSKT08} and consists of first
detecting the heavy hitters (e.g. using sampling), then treating them
differently from the others values, e.g. by partitioning tuples with
heavy hitters on the other attributes; a detailed description is
in~\cite{DBLP:conf/sigmod/XuKZC08}.  None of these algorithms has been
proven to be optimal in any formal sense, and, in fact, there are no
lower bounds for the communication required to compute a join in the
presence of skew.

Complex queries often involve multiple joins, and the traditional
approach is to compute one join at a time leading to a number of
communication rounds at least as large as the depth of the query plan.
It is possible, however, to compute a multiway join in a single
communication round, using a technique that can be traced back to
Ganguli, Silberschatz, and
Tsur~\cite[Sec.7]{DBLP:journals/jlp/GangulyST92}, and was described by
Afrati and Ullman~\cite{DBLP:conf/edbt/AfratiU10} in the context of
MapReduce algorithms.  We will refer to this technique as the {\em
  \textsc{HyperCube} algorithm}, following~\cite{DBLP:conf/pods/BeameKS13}.
The $p$ servers are organized into a hypercube with $k$ dimensions,
where $k$ is the number of variables in the query.  During a single
reshuffling step, every tuple is sent to all servers in a certain
subcube of the hypercube (with as many dimensions as the number of query
variables missing from the tuple).  One challenge in this approach
is to determine the size of the hypercube in each of the $k$ dimensions.
In~\cite{DBLP:conf/edbt/AfratiU10} this is treated as a non-linear
optimization problem and solved using Lagrange multipliers.
In~\cite{DBLP:conf/pods/BeameKS13} it is shown that, in the case where
all relations have the same cardinality, the optimal dimensions
are expressed in terms of the optimal fractional vertex cover of the query.  All
hypercube-based techniques described
in~\cite{DBLP:conf/edbt/AfratiU10,DBLP:conf/pods/BeameKS13}, and
elsewhere (e.g. in~\cite{DBLP:conf/www/SuriV11} for computing
triangles) assume that the data has no skew.  The behavior of the
algorithm on skewed data has not been studied before and no techniques
for addressing skew have been proposed.

\introparagraph{Our Contribution}
In this paper we study the problem of computing a full conjunctive
query (multi-way join query), assuming that the input data may have
arbitrary skew.  We are given $p$ servers that have to compute a query
on a databases with $m$ tuples; we assume $m \gg p$.  We prove
matching upper and lower bounds (up to poly-log $p$ factor) for the
amount of communication needed to compute the query in one
communication round.  We assume the following statistics on the input
database to be known: the cardinality of each input relation, the
identity of the heavy hitters, and their (approximate) frequency in the data. 
We
note that this is a reasonable assumption in today's distributed query
engines.  In our settings there are at most $O(p)$ heavy hitters
because we choose a threshold for the frequency of heavy hitters that
is $\geq m/p$, and therefore the number of heavy hitters is tiny
compared to the size of the database.  We assume that at the beginning
of the computation all servers know the identity of all heavy hitters,
and the (approximate) frequency of each heavy hitter. 
Given these statistics, we
describe an explicit formula and prove that it is both an upper (up to
a poly-log $p$ factor) and a lower bound for the amount of
communication needed to compute the query on the class of databases
satisfying those statistics.  Our results are significant extensions
of our previous results~\cite{DBLP:conf/pods/BeameKS13} which hold
in the absence of skew.

Grohe and Marx~\cite{DBLP:conf/soda/GroheM06} and Atserias et
al.~\cite{DBLP:conf/focs/AtseriasGM08} give upper
bounds on the query size in terms of a fractional {\em edge cover};
this is also a lower bound on the running time of any
sequential algorithm that computes the query.  Recently, Ngo at
al.~\cite{DBLP:conf/pods/NgoPRR12} described a sequential
algorithm that matches that bound.  Thus, the sequential complexity of a
query is captured by the edge cover; our results show that the
communication complexity for parallel evaluation is captured by the
{\em edge packing}.

\introparagraph{Overview of the results} Our analysis of skew starts with
an analysis of skew-free databases, but with unequal cardinalities.
Consider a simple cartesian product, $q(x,y) = S_1(x), S_2(y)$, of two
relations with cardinalities $m_1, m_2$.  Assume $1/p \leq m_1/m_2
\leq p$\footnote{if $m_1 < m_2/p$ then we can broadcast $S_1$ to all servers
and compute the query with a load increase of at most $m_2/p$ per
server thus at most double that of any algorithm, because $m_2/p$ is
the load required to store $S_2$.}.  Let $p_1 = \sqrt{m_1p/m_2}$, $p_2
= \sqrt{m_2p/m_1}$ and assume they are integer values.  Organize the
$p$ servers into a $p_1 \times p_2$ rectangle, and assign to each server
two coordinates $(i,j) \in [p_1] \times [p_2]$.  During the shuffle
phase the algorithm uses two hash functions $h_1,h_2$ and sends every
tuple $S_1(x)$ to all servers with coordinates $(h_1(x), j)$,
$j\in[p_2]$ (thus, every server receives $m_1/p_1 = \sqrt{m_1m_2/p}$
tuples from $S_1$), and sends every tuple $S_2(y)$ to all servers with
coordinates $(i,h_2(y))$, $i \in [p_1]$.  The load per server is $L =
2 \sqrt{m_1m_2/p}$, and it is not hard to see that this is
optimal\footnote{Let $a_i, b_i$ be the number of
  $S_1$-tuples and $S_2$-tuples received by server $i\in[p]$.  On
  one hand $\sum_i a_ib_i = \iprod{\bar{a},\bar{b}} \geq m_1m_2$
  because the servers must report all $m_1m_2$ tuples; on the
  other hand $\iprod{\bar{a},\bar{b}} \leq \norm{\bar{a}+\bar{b}}_2^2/4 
  \leq p \norm{\bar{a}+\bar{b}}_\infty^2/4
  = p L^2/4$.}.  This
observation generalises to any $u$-way cartesian product: the minimum
load per server needed to compute $S_1 \times \ldots \times S_u$ is
$\Omega((m_1 m_2 \cdots m_u/p)^{1/u})$

Consider now some arbitrary full conjunctive query $q$ over relations
$S_1, \ldots, S_\ell$, and assume that the cardinality of $S_j$ is
$m_j$.  Choose some subset $S_{j_1}, S_{j_2}, \ldots, S_{j_u}$; the
subset is called an {\em edge packing}, or an {\em edge matching}, if
no two relations share a common variable.  Any one-round algorithm
that computes the query correctly must also compute the cartesian
product of the relations $S_{j_1}, S_{j_2}, \ldots$ Indeed, since no
two relations share variables, any tuple in their cartesian product
could potentially be part of the query answer; without knowing the
content of the other relations, the input servers that store
(fragments of) $S_{j_1}, S_{j_2}, \ldots$ must ensure that any
combination reaches some output server.  Therefore, the load per
server of any one-round algorithm is at least $\Omega((m_{j_1} \cdots
m_{j_u} / p)^{1/u})$.  Thus, every edge packing gives a lower bound for
computing $q$.  For example the load per server needed to compute
$q(x,y,z,w) = S_1(x,y),S_2(y,z),S_3(z,w)$ is at least $L \geq
\sqrt{m_1m_3/p}$, because of the packing $\set{S_1, S_3}$; it must
also be $L \geq m_2/p$, because of the packing $\set{S_2}$.  We prove
in this paper that this property extends to any {\em fractional edge
  packing}.  Denote $M_j$ the number of bits needed to represent the
relation $S_j$.  We show:

\begin{theorem} \label{th:no-skew} Let $\bu = (u_1, \ldots, u_\ell)$ 
be any fractional edge packing for the query $q$,
  and $u = \sum_j u_j$.  Denote $K(\bu,\bM) = \prod_j M_j ^ {u_j}$ and
  $L(\bu,\bM,p) = \left(K(\bu,\bM) / p\right)^{1/u}$.  If an algorithm
  computes $q$ in one step, then at least one server has a
  load $\Omega(L(\bu,\bM,p))$.  Conversely, let $L_\texttt{lower} =
  \max_{\bu} L(\bu,\bM,p)$ be the maximum over all fractional edge
  packings.  Then, there exists a randomized algorithm for $q$ (
  the \textsc{HyperCube} algorithm, HC) whose maximum load per server is
  $O(L_\texttt{lower} \ln^{k} p)$ with high probability on all
  database without skew.
\end{theorem}

In the case when all relations have the same size $M$, then the lower
bound is $L(\bu, \bM, p)= M/p^{1/u}$, whose maximum value is
$M/p^{1/\tau^*}$ where $\tau^*$ is the value of the maximal fractional
edge packing, and is equal to the fractional vertex covering number
for $q$; thus, we recover our prior result
in~\cite{DBLP:conf/pods/BeameKS13}, which was stated for the special
case when all relations are {\em matchings}.


\autoref{th:no-skew} completes the analysis of the HC algorithm on
skew-free databases with arbitrary cardinalities.  In addition, we
prove a rather surprising result: the HC algorithm is resilient to
skew, in the sense that, even on skewed databases, it can still offer
a non-trivial upper bound for the maximum load per server: namely $L =
O(m/p^{1/k})$, where $m$ is the largest cardinality, and $k$ the total
number of variables in the query.  For example, using HC one may
compute the join of two relations and guarantee a load of
$O(m/p^{1/3})$, {\em even without any knowledge about skew or heavy
  hitters}.  In contrast, a standard hash-join algorithm may incur a
load of $\Omega(m)$ when the join attributes have a single
value.

Next, we consider the case when information about heavy hitters is
known.  In addition to knowing the cardinalities of the input
relations, we assume that the identities of the heavy hitters are
known, and that the frequency in the data of every heavy hitter is
also known.  For example, if the relation $S_j$ contains an attribute
$x$, then we assume to know the set of heavy hitters $H$, together
with the frequencies $m_j(h) = |\sigma_{x=h}(S_j)|$, which, by
definition, are $\geq m_j/p$.

For this setting we generalize the results for skew-free by proving a
lower bound, and a matching upper bound.  Our lower bound is an
elegant generalization of that in \autoref{th:no-skew}, and is
expressed in terms of fractional edge packings of {\em residual
  queries}: for each set of variables $\bx$, the residual query
$q_\bx$ is obtained from $q$ by simply removing the variables $\bx$.
Our matching upper bound is based on the idea of running the main
query on the subset of the database that consists of light hitters,
then handling each heavy hitter separately, by computing a residual
query.  The algorithm is difficult, because of two challenges.  First,
one needs to consider sets of attributes of each relation $S_j$ that
may be heavy hitters jointly, even if none of them is a heavy hitter
by itself.  Second, an attribute value may become a heavy hitter in
the residual query even though it was light in the main query.  Our
algorithm addresses these challenges by creating, for each subset of
attributes of each relation, $O(\log p)$ bins of heavy hitters, where
all heavy hitters in a bin have frequencies that differ by at most a
factor of two.  (Because of this it suffices for our algorithm to have
access only to approximate frequencies of heavy hitters.)  By
considering separately all combinations of bins, we can run residual
queries on databases where the frequencies are guaranteed to be
uniform, thus avoiding the difficulties that arise from recursion.
Denote $M_j(h)$ the number of bits needed to represent the subset
$\sigma_{x=h}(S_j)$ of $S_j$.  Our second main result is:

%
%
%
%

\begin{theorem}
  Consider all database instances defined by a set of statistics
  consisting of the cardinalities of the relations, the set of heavy
  hitters, and the frequency of each heavy hitter.  For a set of
  variables $\bx$ and any packing $\bu$ of the residual query $q_\bx$
  that saturates the variables in $\bx$, let $L_\bx(\bu, \bM, p) =
  \left(\sum_\bh K(\bu, \bM(\bh)) / p\right)^{1/u}$: then any
  deterministic algorithm that computes $q$ on these
  databases must have a load $\geq L_\bx(\bu, \bM, p)$.  Moreover,
  denoting $L_{\texttt{lower}} = \max_{\bx, \bu} L_\bx(\bu, \bM, p)$,
  there exists a randomized algorithm for computing $q$ whose load per
  server is $O(L_{\texttt{lower}} \log^{O(1)} p)$ with high
  probability.
\end{theorem}

As a final contribution of our paper, we discuss the connection between
our results in the MPC model and the results of~\cite{ASSU13} on
models for computation in MapReduce. We show that our results provide
new upper and lower bounds for computing conjunctive queries in~\cite{ASSU13},
and in a stronger computational model.

The paper is organized as follows.  We give the background in
\autoref{sec:basic}, then present in \autoref{sec:basic} our results
for the case when the statistics known about the database are
restricted to cardinalities.  The case of databases with known heavy
hitters is discussed in \autoref{sec:algo}. We present the connection with~\cite{ASSU13} in
\autoref{sec:map-reduce} and finally conclude in
\autoref{sec:conclusions}.  Several proofs are relegated to the appendix.

%% file: model.tex
\section{Preliminaries}
\label{sec:defs}

We review here the basic definitions
from~\cite{DBLP:conf/pods/BeameKS13}.

\subsection{Massively Parallel Communication}
\label{subsec:model}

We define here the \mpc\ model.  The computation is performed by $p$
servers, called {\em workers}, connected by a complete network of
private\footnote{``Private'' means that when server $i$ sends a
  message to server $j$ no other server sees its content.} channels.
The input data is initially distributed evenly among the $p$ workers.
The computation proceeds in rounds, where each round consists of local
computation at the workers interleaved with global communication. The
servers have unlimited computational power, but may be limited in the
amount of bits they receive.  In this paper, we discuss query
evaluation in this model, and consider a single round of
communication.  The {\em load} of a server is the number of bits received
by the server during the communication; we write $L$ for the maximum
load among all servers.

If the size of the input is $M$ bits, an ideal algorithm would split
the data equally among the servers, and so we would like to have $L =
O(M/p)$; in this case, the total amount of data communicated is $O(M)$
and thus there is no {\em replication}.  Depending on the query, $L$ is
higher than the ideal $M/p$ by some factor called {\em replication
  factor}.  In~\cite{DBLP:conf/pods/BeameKS13} we considered the case
when the input data is perfectly uniform and showed that the
replication factor for any conjunctive query is $O(p^\varepsilon)$,
where $0 < \varepsilon \leq 1$ is a constant that depends only on the
query. In this work we consider arbitrary input data,
and the replication factor will be a more complex formula that depends
on the database statistics.

%
%


\introparagraph{Randomization}
The \mpc\ model allows randomization during the computation. The random 
bits are available to all servers, and are computed independently of the input
data. 

\introparagraph{Random Instances and Yao's Lemma} Our lower bounds
are stated by showing that, if the database instance is chosen at
random from some known probability space, then any algorithm with a
load less than a certain bound can report only $o(1)$ fraction of the
expected number of answers to the query.  Using Yao's
lemma~\cite{yao83} this implies than for any randomized algorithm
there exists an instance on which the algorithm will fail with
high probability; we refer to~\cite{DBLP:conf/pods/BeameKS13} for
details.

\introparagraph{Input Servers} In our upper bounds we assume that the
input relations $S_j$ are initially partitioned uniformly on the
servers: all our algorithm treat tuples in $S_j$ independently of
other tuples.  For our lower bounds, we assume a more powerful model,
where, at the beginning of the algorithm, each relation $S_j$ is
stored on a separate server, called an {\em input server}, which can
examine the entire relation in order to determine what message to
send.  These assumptions are the same as
in~\cite{DBLP:conf/pods/BeameKS13}.  

\introparagraph{Database Statistics} In this paper we assume that all
input servers know certain database statistics.  {\em Simple database
  statistics} consists of the cardinalities $m_j$ of all input
relations $S_j$; we discuss this case in \autoref{sec:basic}.  {\em
 Complex database statistics} add information about heavy hitters; we
discuss these in the rest of the paper.  The size of these statistics
is $O(1)$ in the first case, and $O(p)$ in the second.  Both upper and
lower bounds assume that these statistics are available to all input
servers.

%

\subsection{Conjunctive Queries}

\label{subsec:cq}

We consider computing answers to conjunctive queries
over an input database in the \mpc\ model.  We fix an input vocabulary
$S_1, \ldots, S_\ell$, where each relation $S_j$ has arity
$a_j$; let $a = \sum_{j =1}^{\ell} a_j$.  The input data
consists of one relation instance for each symbol.
We consider full conjunctive queries  without
self-joins:
\begin{equation} \label{eq:q}
  q(x_1,\ldots, x_k) = S_1(\bar x_1), \ldots, S_\ell(\bar x_\ell) 
\end{equation}
The query is {\em full}, meaning that every variable in the body
appears the head (for example $q(x) = S(x,y)$ is not full), and {\em
  without self-joins}, meaning that each relation name $S_j$ appears
only once (for example $q(x,y,z) = S(x,y), S(y,z)$ has a
self-join). The {\em hypergraph} of a query $q$ is defined by
introducing one node for each variable in the body and one hyperedge
for each set of variables that occur in a single atom.  With some
abuse we write $i \in S_j$ to mean that the variable $x_i$ occurs in
the the variables $vars(S_j)$ of the atom $S_j$.


\introparagraph{Fractional Edge Packing}
A {\em fractional edge packing} (also known as a {\em fractional
  matching}) of a query $q$ is any feasible solution
$\mathbf{u} = (u_1, \dots, u_{\ell})$ of the following linear
constraints:
\begin{align}
  & \forall i \in [k]:
  \sum_{j: i \in S_j} u_j \leq 1 \label{eq:cover:dual} \\
  & \forall j \in [\ell] : u_j \geq 0 \nonumber
\end{align}

The edge packing associates a non-negative weight $u_j$ to each
atom $S_j$ s.t. for every variable $x_i$, the sum of the weights for
the atoms that contain $x_i$ do not exceed 1. If all inequalities are
satisfied as equalities by a solution to the LP, we say that the
solution is {\em tight}.

For a simple example, an edge packing of the query $L_3 = S_1(x_1,
x_2),S_2(x_2,x_3), S_3(x_3, x_4)$ is any solution to $u_1 \leq 1$,
$u_1+u_2 \leq 1$, $u_2 + u_3 \leq 1$ and $u_3 \leq 1$. In particular,
the solution $(1,0,1)$ is a tight and feasible edge packing.
A {\em fractional edge cover} is a feasible solution $\mathbf{u} =
(u_1, \dots, u_{\ell})$ to the system above where $\leq$ is replaced
by $\geq$ in Eq.\ref{eq:cover:dual}.  Every tight fractional edge packing
is  a tight fractional edge cover, and vice versa.

\subsection{Friedgut's Inequality}

Friedgut~\cite{friedgut2004hypergraphs} introduces the following class
of inequalities.  Each inequality is described by a hypergraph, which
in our paper corresponds to a query, so we will describe the
inequality using query terminology.  Fix a query $q$ as in
\eqref{eq:q}, and let $n > 0$.  For every atom $S_j(\bar x_j)$ of
arity $a_j$, we introduce a set of $n^{a_j}$ variables $w_{j}(\ba_j)
\geq 0$, where $\ba_j \in [n]^{a_j}$. If $\ba \in [n]^k$, we denote by
$\ba_j$ the vector of size $a_j$ that results from projecting on the
variables of the relation $S_j$.
Let $\mathbf{u} = (u_1, \dots, u_{\ell})$ be a fractional {\em edge
  cover} for $q$.  Then:
\begin{align}
  \sum_{\ba \in [n]^k} \prod_{j=1}^{\ell} w_{j}( \ba_j) \leq & \prod_{j=1}^{\ell}
  \left(\sum_{\ba_j \in [n]^{a_j}}  w_{j} (\ba_j)^{1/u_j}\right)^{u_j} \label{eq:friedgut}
\end{align}
We illustrate Friedgut's inequality on $C_3$:
\begin{align}
C_3(x,y,z) = S_1(x,y),S_2(y,z),S_3(z,x) \label{eq:ccc}
\end{align}
$C_3$ has cover $(1/2,1/2,1/2)$.  Thus, we
obtain the following, where $a,b,c$ stand for
$w_1,w_2,w_3$ respectively:
\begin{align*}
  \sum_{x,y,z \in [n]}\kern -1em a_{xy}\cdot b_{yz} \cdot c_{zx} \leq &
  \sqrt{\sum_{x,y \in [n]} a_{xy}^2 \sum_{y,z \in [n]} b_{yz}^2
    \sum_{z,x \in [n]} c_{zx}^2} 
\end{align*}

Friedgut's inequalities immediately imply a well known result
developed in a series of
papers~\cite{DBLP:conf/soda/GroheM06,DBLP:conf/focs/AtseriasGM08,DBLP:conf/pods/NgoPRR12}
that gives an upper bound on the size of a query answer as a function
on the cardinality of the relations.  For example in the case of
$C_3$, consider an instance $S_1, S_2, S_3$, and set $a_{xy} = 1$ if
$(x,y) \in S_1$, otherwise $a_{xy}=0$ (and similarly for
$b_{yz},c_{zx}$).  We obtain then $ |C_3| \leq \sqrt{|S_1| \cdot |S_2|
  \cdot |S_3|}$.

%% file: section2-basic.tex
\section{Simple Database Statistics}
\label{sec:basic}

In this section we consider the case when the statistics on database
consists of the cardinalities $m_1, \ldots, m_\ell$ of the relations
$S_1, \ldots, S_\ell$. All input servers know these statistics.  We
denote $\bm= (m_1, \ldots, m_\ell)$ the vector of cardinalities, and $\bM
= (M_1, \ldots, M_\ell)$ the vector of the sizes expressed in bits,
where $M_j = a_j m_j \log n$, and $n$ is the size of the domain of each
attribute.

\subsection{The HyperCube Algorithm}

We present here the \textsc{HyperCube} (HC) algorithm and its
analysis.

The HC algorithm, first described
in~\cite{DBLP:conf/edbt/AfratiU10}, expresses the number of servers $p$
as $p=p_1 \cdot p_2 \cdots p_k$, 
where each $p_i$ is called the {\em share} for the variable
$x_i$.  The algorithm uses $k$ independently chosen random hash
functions $h_i : [n] \rightarrow [p_i]$, one for each variable $x_i$.
During the communication step, the algorithm sends every tuple
$S_j(\ba_j) = S_j(a_{i_1}, \dots, a_{i_{r_j}})$ to all servers
$\mathbf{y} \in [p_1] \times \dots \times [p_k]$ such that
$h_{i_m}(a_{i_m}) = \mathbf{y}_{i_m}$ for any $1 \leq m \leq r_j$.  In
other words, for every tuple in $S_j$, after applying the hash
functions the algorithm knows the coordinates for the dimensions $i_1,
\ldots, i_{r_j}$ in the hypercube, but does not know the other
coordinates, and it simply replicates the tuple along those other
dimensions.  The algorithm finds all answers, because each potential
output tuple $(a_1, \ldots, a_k)$ is known by the server $\mathbf{y} =
(h_1(a_1), \dots, h_k(a_k))$.  

Since the HC algorithm is parametrized by the choice of shares, we next address
two issues. First, we choose the shares $p_i$ so as to minimize the expected 
load per server. Second, we prove that, with high probability on the choices of the random hash
functions, the expected load is not exceeded by more than a factor for any server.
We start with the latter, which was not addressed
in~\cite{DBLP:conf/edbt/AfratiU10}, and was addressed only in a
limited setting in~\cite{DBLP:conf/pods/BeameKS13}: our analysis
reveals a previously unknown property of the HC algorithm.

\introparagraph{Analysis of the Load Per Server} Our analysis is based
on the following lemma about hashing.

\begin{lemma}
  \label{lemma:hashing} 
  Let $R(A_1, \dots, A_r)$ be a relation of arity $r$ with at most $m$
  tuples. Let $p_1, \ldots, p_r$ be integers and denote $p = \prod_i
  p_i$ where $m\ge p$. Suppose that we hash each tuple $(a_1,
  \ldots, a_r)$ to the bucket $(h_1(a_1), \ldots, h_r(a_r))$, where
  $h_1, \ldots, h_r$ are independent and perfectly random hash 
  functions. Then:
  \begin{packed_enum}
  \item The expected load in every bucket is $m/p$.
  \item If for every $i\in[r]$, every value of the attribute
    $A_i$ occurs at most once, then the maximum load per bucket is
    $O(m /p)$ with high probability.
  \item If for every $S\subseteq [r]$, every tuple of values of attributes
    $(A_i)_{i\in S}$ occurs at most $am/\prod_{i\in S} p_i$ times, then the
    maximum load per bucket is $O(a^r\ln^r(p) m/p)$ with high probability.
  \item The maximum load per bucket is $O(m / \min_i (p_i))$,
    independent of the instance.  
   \end{packed_enum}
\end{lemma}

We prove this lemma in \autoref{sec:hashing}. For the case where the arity
of the relation is $r=1$, the above lemma is a straightforward application of Chernoff bounds.
The case where $r \geq 2$ requires a more sophisticated argument and
is novel, to the best of our knowledge.   

We can apply the lemma to
analyze the behavior of the HC algorithm under two conditions: over skew-free
databases, and over arbitrary databases.  Given a vector of shares
$(p_1, \dots, p_k)$, we say that a relation $S_j$ is {\em skew-free}
w.r.t. the shares if for every subset of variables $\bx \subseteq vars(S_j)$,
every value has frequency at most $m_j/\prod_{x_i\in \bx} p_i$. 
Our prior analysis
in~\cite{DBLP:conf/pods/BeameKS13} was only for the special case when
the frequency of each value at each attribute is at most 1.

\begin{corollary}
\label{cor:hashing}
Let $\mathbf{p} = (p_1, \dots, p_k)$ be the shares of the HC algorithm.

  (i) If $S_j$ is skew-free w.r.t. $\mathbf{p}$, then with high
  probability the maximum load per server  is
  $$O \left( \max_j \frac{M_j}{\prod_{i: i \in S_j} p_i}  \ln^{k}(p) \right)$$ 
  (ii) For any given database, with high probability the maximum load
  per server is
  $$O \left( \max_j \frac{M_j}{\min_{i: i \in S_j}(p_i)}  \right)$$
\end{corollary}

In~\cite{DBLP:conf/edbt/AfratiU10}, it is assumed that the database is
skew-free and that the load per server is the expected load; item (i)
of our result confirms that the load does not exceed the expected load
by more than a poly-log factor with high probability (our proof uses
Chernoff bounds).  Item (ii) is novel, because it describes how the HC
algorithm behaves on skewed data: it shows that the algorithm is
resilient to skew, and gives an upper bound even on skewed databases.
We illustrate with an example.

\begin{example}
  Let $q(x,y,z) = S_1(x,z),S_2(y,z)$ be a simple join, where both
  relations have cardinality $m$.  We show two instances of the HC 
  algorithm, the first optimized for skewed databases, and the second optimized for
  skew-free databases.  The first share allocation is $p_1 = p_2 = p_3 =
  p^{1/3}$, thus every processor is identified by $(w_1,w_2,w_3) \in
  [p_1] \times [p_2] \times [p_3]$.  The algorithm sends every tuple
  $S_1(a,c)$ to all processors $(h_1(a),w_2, h_3(c))$ for $w_2 \in
  [p_3]$ and every tuple $S_2(b,c)$ to all processors
  $(w_1,h_2(b),h_3(c))$ for $w_1 \in [p_1]$.  By
  \autoref{cor:hashing}, on skew-free databases the load per server is
  $O(m/p^{2/3})$ (plus polylog factor).  But even on skewed database
  the load per server is $O(m/p^{1/3})$.  The second algorithm
  allocates shares $p_1=p_2 = 1, p_3 = p$.  This corresponds to a
  standard hash-join on the variable $z$.  On a skew-free database
  (equivalently, when every value of $z$ has frequency at most $m/p$ in both relations)
  the load per server is $O(m/p)$ with high probability. However, if it is skewed,
  then the load can be as bad as $O(m)$: this occurs when all tuples have
  the same value $z$.
\end{example}

Generalizing the example, for every conjunctive query with $k$
variables, we can execute the HC algorithm with equal shares $p_1 = \ldots = p_k =
p^{1/k}$. Then, the algorithm achieves a maximum load per server of at most
$O(\max_j M_j / p^{1/k})$. 

However,
in practice, in applications where skew is expected, it is better to
design specialized algorithms, as we further discuss in~\autoref{sec:algo}.
Therefore, we focus our analysis on skew-free databases, and optimize
the expected load.

\introparagraph{Choosing the Shares} Here we discuss how to compute
the shares $p_i$ to optimize the expected load per server.  Afrati and
Ullman compute the shares by optimizing the total load $\sum_j
m_j/\prod_{i: i \in S_j} p_i$ subject to the constraint $\prod_i p_i =
1$, which is a non-linear system that can be solved using Lagrange
multipliers.  Here we take a different approach.  First, we write the
shares as $p_i = p^{e_i}$ where $e_i \in [0,1]$ is called the {\em
  share exponent} for $x_i$, and denote $L$ the maximum load per
server, thus $M_j/\prod_{i: i \in S_j} p_i \leq L$ for every $j$.
Denote $\lambda = \log_p L$ and $\mu_j = \log_p M_j$ 
(we will assume w.l.o.g. that $m_j \geq p$, hence $\mu_j \geq 1$ for all $j$).
Then, we optimize the LP:
\begin{align}
\text{minimize}   \quad & \lambda \nonumber \\
\text{subject to} \quad
                  &  \sum_{i \in [k]} -e_i \geq -1 \nonumber \\
\quad \forall j \in [\ell]:  & \sum_{i \in S_j} e_i + \lambda \geq \mu_j \nonumber \\
\quad \forall i \in [k]: & e_i \geq 0, \quad \lambda \geq 0 \label{eq:primal:lp}
\end{align}

Denote $L_{\texttt{upper}} = p^{e^*}$ where  $e^*$ is the objective value of the optimal
solution to the above LP.  We have:

\begin{theorem}
  For a query $q$ and $p$ servers, with statistics $\bM$, let $\be =
  e_1, \dots, e_k$ be share exponents that are optimal for the above
  LP. Then, the expected load per server is $L_{\texttt{upper}}$.
  Moreover, if every $S_j$ is skew-free w.r.t. to $\be$, then the
  maximum load per server is $O(L_{\texttt{upper}} \cdot \ln^k(p))$
  with high probability.
\end{theorem}

In~\autoref{subsec:lower:eq:upper} we will give a closed form
expression for $L_{\texttt{upper}}$ and also provide an example.
But first, we prove a matching lower bound.

\subsection{The Lower Bound}

We next prove a lower bound for the maximum load per server over
databases with statistics $\bM$.  Fix some constant $0 < \delta < \min_j
\{ a_j \}$, and assume that for every relation $S_j$, its cardinality
satisfies $m_j \leq n^{\delta}$, where $n$ is the domain size of each attribute.

Consider the probability space where each relation $S_j$ is chosen
independently and uniformly at random from all subsets of $[n]^{a_j}$
with exactly $m_j$ tuples.  Denote $\E[|q(I)|]$ the expected number of
answers to $q$.  (We show in the appendix, \autoref{lem:expected_size}, that
$\E[|q(I)|] = n^{k-a} \prod_{j=1}^{\ell} m_j$.)

Fix a query $q$ and a fractional edge packing $\bu$ of $q$.  Denote $u =
\sum_{j=1}^{\ell} u_j$ the value of the packing, and:
\begin{align}
  K(\bu,\bM) = & \prod_{j=1}^{\ell} M_j^{u_j}  \label{eq:kum} \\
  L(\bu,\bM,p) = & \left(\frac{K(\bu,\bM)}{p} \right)^{1/u} \label{eq:lump}
\end{align}
Further denote $L_{\texttt{lower}} = \max_{\bu} L(\bu,\bM,p)$, where 
$\bu$ ranges over all edge packings for $q$. Let  $c$ be a constant,
$c = \frac{a_j-\delta}{3a_j}$, where $a_j$ is the maximum arity of all
relations. We prove in \autoref{sec:lower:uniform}:

\begin{theorem} 
  \label{th:lower:uniform} 
  Fix statistics $\bM$, and consider any deterministic MPC algorithm
  that runs in one communication round on $p$ servers.  (1) Let $\bu$
  be any edge packing of $q$.  If $i$ is any server and $L_i$ is its
  load, then server $i$ reports at most
  $$
  \frac{L_i^u}{c^u K(\bu, \bM)} \E[|q(I)|]
  $$
  answers in expectation, where $I$ is a randomly chosen database with
  statistics $\bM$.  Therefore, the $p$ servers of the algorithm
  report at most
   $$  \left( \frac{L}{c \cdot L(\bu,\bM,p)} \right)^{u} 
   \cdot \E[|q(I)|]$$
   answers in expectation, where $L$ is the maximum load of all servers.

   As a consequence, any algorithm that computes $q$ correctly over
   databases with statistics $M$ must have load $L \geq c
   L_{\texttt{lower}}$ bits\footnote{This follows by observing that,
     when $L(\bu,\bM,p)$ is maximized, then $u = \sum_j u_j \geq 1$.}.
\end{theorem}


In our previous work~\cite{DBLP:conf/pods/BeameKS13}, we presented a
matching lower/upper bound for computing $q$ on some restricted
database instances, where the relations $S_j$ are {\em matchings} and
have the same cardinalities; the proof of \autoref{th:lower:uniform}
is an extension of the lower bound proof
in~\cite{DBLP:conf/pods/BeameKS13}.  We explain here the relationship.
When all cardinalities are equal, $M_1 = \ldots = M_\ell = M$, then
$L_{\texttt{lower}}= M/p^{1/u}$, and this quantity is maximized when
$\bu$ is a maximum fractional edge packing, whose value is denoted $\tau^*$: this is
equal to the fractional vertex covering number for $q$.  The bound
in~\cite{DBLP:conf/pods/BeameKS13} is $M/p^{1/\tau^*}$.  (The constant
$c$ in~\cite{DBLP:conf/pods/BeameKS13} is tighter.)
\autoref{th:lower:uniform} generalizes the lower bound to arbitrary
cardinalities and, in that case, $L(\bu,\bM,p)$ is not necessarily
maximized at $\tau^*$.

In the rest of this section we prove that $L_{\texttt{lower}} =
L_{\texttt{upper}}$ and also give a closed form expression for both.

\subsection{Proof of Equivalence}

\label{subsec:lower:eq:upper}

If $\bu, \bu'$ are two fractional edge packings for a conjunctive
query $q$, then we write $\bu \leq \bu'$ if $u_j \leq u_j'$ for all
$j$ and say that $\bu'$ {\em dominates} $\bu$.  Let $pk(q)$ be the
non-dominated vertices of the polytope defined by the fractional edge
packing constraints in~\eqref{eq:cover:dual}.  Recall that the
vertices of the polytope are feasible solutions $\bu_1, \bu_2, \ldots$,
with the property that every other feasible solution $\bu$ to the LP
is a convex combination of these vertices.  Each vertex can be
obtained by choosing $m$ out of the $k+\ell$ inequalities in
\eqref{eq:cover:dual}, transforming them into equalities, then
solving for $\bu$: $pk(q)$ contains the subset of vertices that are
not dominated by others, and $|pk(q)| \leq {k+\ell \choose m}$.  We
prove here:

\begin{theorem} \label{th:lump} For any vector of statistics $\bM$ and number of
  processors $p$ , we have:
  \begin{align*}
    L_{\texttt{lower}} = L_{\texttt{upper}} = \max_{\bu \in pk(q)}  L(\bu,\bM,p)
  \end{align*}
\end{theorem}

Before we prove the theorem we discuss its implications.  We start with
an example.

\begin{example}
\label{ex:triangle}
  Consider the triangle query 
  $$C_3 = S_1(x_1,x_2), S_2(x_2,x_3),S_3(x_3,x_1)$$ 
  and assume the three cardinalities are $m_1, m_2,
  m_3$.  Then, $pk(C_3)$ has four vertices, and each gives a different value for
  $L(\bu,\bM,p)$:
\renewcommand{\arraystretch}{1.0}
\begin{center}  
\[ \begin{array}{|c|c|}
\hline
  \bu & L(\bu, \bM, p)  \\ \hline
  (1/2, 1/2, 1/2) & (M_1M_2M_3)^{1/3} / p^{2/3}   \\ \hline
   (1, 0, 0) & M_1 / p \\ \hline
(0, 1, 0) & M_2 / p \\ \hline
(0, 0, 1) & M_3 / p \\ \hline
\end{array} \]
\end{center}

  The first vertex is the solution to $u_1+u_2=u_1+u_3=u_2+u_3=1$; the
  second the solution to $u_1+u_2=1, u_2=u_3=0$, etc.  Thus, the load
  of the algorithm is the largest of these four quantities, and this
  is also the lower bound of any algorithm.  In other words, the
  optimal solution to the LP (\ref{eq:primal:lp}) can be given in
  closed form, as the maximum over four expressions.
\end{example}

Next we use the theorem to compute the {\em space exponent}.  In
~\cite{DBLP:conf/pods/BeameKS13} we showed that, for every query $q$,
the optimal load over databases restricted to matchings is
$M/p^{1-\varepsilon}$, where $0 \leq \varepsilon < 1$ is called the space
exponent for $q$.  Consider now a database with arbitrary statistics,
and denote $M = \max_j M_j$.  If for some $j$, $M_j \leq M/p$, then an
optimal algorithm can simply broadcast $M_j$ to all $p$ servers, and
remove $S_j$ from the query: this will increase the load by at most a
factor of 2.  Thus, we may assume w.l.o.g. that for every $j$, $M_j =
M/p^{\nu_j}$ for some $\nu_j \in [0,1)$.  Then $L(\bu,\bM,p) = M /
p^{\sum_j \nu_ju_j + 1/u}$.  To obtain the optimal load, one needs to
find $\bu \in pk(q)$ that minimizes $v = \sum_j \nu_ju_j + 1/(\sum_j
u_j)$.  Denoting $v^*$ the minimal value, the load is $M/p^{v^*}$.
Thus, the space exponent for given statistics is $1-v^*$.

\begin{proof} (of \autoref{th:lump}) Recall that $L_{\texttt{upper}}$
  is $p^{e^*}$, where $e^*$ is the optimal solution to the {\em
    primal} LP problem \eqref{eq:primal:lp}.  Consider its {\em
    dual} LP:
\begin{align}
  \text{maximize}   \quad & \sum_{j \in [\ell]} \mu_j f_j - f \nonumber \\
  \text{subject to} \quad
  &  \sum_{j \in [\ell]} f_j \leq 1 \nonumber \\
  \quad \forall i \in [k]: & \sum_{j:i \in S_j} f_j -f \leq 0 \nonumber \\
  \quad \forall j \in [\ell]: & f_j \geq 0, \quad f \geq 0  \label{eq:dual:lp}
\end{align}

By the primal/dual theorem, its optimal solution is also $e^*$.
Writing $u_j = f_j/f$ and $u = 1/f$, we transform it into the
following non-linear optimization problem:
\begin{align}
  \text{maximize}   \quad &  \frac{1}{u} \cdot  \left( \sum_{j \in [\ell]}\mu_j u_j - 1 \right) \nonumber \\
  \text{subject to} \quad
  &  \sum_{j \in [\ell]} u_j \leq u \nonumber \\
  \quad \forall i \in [k]: & \sum_{j:i \in S_j} u_j \leq 1 \nonumber \\
  \quad \forall j \in [\ell]: & u_j \geq 0 \label{eq:dual:lp:2}
\end{align}

Consider optimizing the above non-linear problem.  Its optimal
solution must have $u = \sum_j u_j$, otherwise we simply replace $u$
with $\sum_j u_j$ and obtain a feasible solution with at least as good
objective function (indeed, $\mu_j \geq 1$ for any $j$, and hence 
$\sum_j \mu_j u_j \geq \sum_j u_j \geq 1$, since any optimal $\bu$ will have sum at least 1).
Therefore, the optimal is given by a fractional
edge packing $\bu$.  Furthermore, for any packing $\bu$, the objective
function $\sum_j \frac{1}{u} \cdot (\mu_j u_j - 1)$ is $\log_p
L(\bu, \bM, p)$.  To prove the theorem, we show that (a) $e^* =
u^*$ and (b) the optimum is obtained when $\bu \in pk(q)$.  This
follows from:

\begin{lemma} \label{lem:convex:mapping}
  Consider the function $F : \mathbf{R}^{k+1} \rightarrow
  \mathbf{R}^{k+1}$: $F(x_0, x_1, \ldots, x_k) = (1/x_0, x_1/x_0,
  \ldots, x_k/x_0)$.  Then:
  \begin{packed_item}
  \item $F$ is its own inverse, $F = F^{-1}$.
  \item $F$ maps any feasible solution to the system
    \eqref{eq:dual:lp} to a feasible solution to
    \eqref{eq:dual:lp:2}, and conversely.
  \item $F$ maps a convex set to a convex set.
  \end{packed_item}
\end{lemma}

\begin{proof}
  If $y_0 = 1/x_0$ and $y_j = x_j/x_0$, then obviously $x_0 = 1/y_0$
  and $x_j = y_j/y$.  The second item can be checked directly.  For
  the third item, it suffices to prove that $F$ maps a convex
  combination $\lambda \mathbf{x} + \lambda' \mathbf{x}'$ where
  $\lambda+\lambda'=1$ into a convex combination $\mu F(\mathbf{x}) +
  \mu' F(\mathbf{x}')$, where $\mu+\mu'=1$.  Assuming $\mathbf{x} =
  (x_0, x_1, \ldots, x_k)$ and $\mathbf{x}' = (x_0', x_1', \ldots,
  x_k')$, this follows by setting $\mu = x_0 / (\lambda x_0 + \lambda
  x_0')$ and $\mu' = x_0' / (\lambda x_0 + \lambda x_0')$.
\end{proof}

This completes the proof of \autoref{th:lump}.
\end{proof}

%% file: section3-algorithm.tex
\section{Complex Database Statistics}
\label{sec:algo}

In this section, we discuss algorithms and lower bounds in the case
where the input servers are provided by additional information
regarding heavy hitters.

\subsection{A Simple Case: Join}

\label{subsec:join}

We start with a simple example, the join of two tables, $q(x,y,z)
= S_1(x,z), S_2(y,z)$, to illustrate the main algorithmic and proof
ideas.  The algorithm uses the same principle popular in virtually all
parallel join implementations to date: identify the heavy hitters and
treat them differently.  However, the analysis and optimality proof is
new, to the best of our knowledge.

Let $m_1, m_2$ be the cardinalities of $S_1, S_2$.  For any value $h \in [n]$
that variable $z$ may assume, let $m_j(h)$ denote the frequency of $h$ in
$S_j$, $j=1,2$; 
$h$ is called a {\em heavy hitter} in $S_j$ if $m_j(h) \ge m_j/p$.  Let
$H_{12}$ denote the set of all heavy hitters in both $S_1$ and $S_2$,
$H_1$ those heavy only in $S_1$, and $H_2$ those heavy only in $S_2$.
Note that $|H_{12}|, |H_1|, |H_2| \leq p$.  We assume that all heavy
hitters and their frequencies are known by the algorithm.  A tuple
$S_j(a,b)$, for $j=1,2$, is called a {\em light hitter} if the value $b$ 
is in none of the sets $H_1, H_2, H_{12}$.

For each $h \in H_{12}$, we must compute a cartesian product
$q[h / z] = S_1(x,h),S_2(y,h)$: if we allocate $p_h$ servers to this
task, then we have seen in Section~\ref{sec:intro} that the optimal way is
to write $p_h=p_1 \cdot p_2$ where $p_1 = \lceil p_h
\sqrt{m_1(h)/m_2(h)} \rceil$ and $p_2 = \lceil p_h
\sqrt{m_2(h)/m_1(h)} \rceil$, then use the algorithm in
Section \ref{sec:intro}.  However, if $h \in H_1$, we cannot apply
the same formula because $m_1(h) \gg p_h m_2(h)$ and we may obtain
$p_1 \gg p_h$: in that case we allocate all servers to the variable
$x$.  Define:
\begin{align*}
  & \forall h \in H_{12}: \quad  K_{12}(h) = m_1(h) \cdot m_2(h)\\
  & p_{h} = \left\lceil p \frac{K_{12}(h)}{\sum_{h' \in H_{12}} K_{12}(h')}  \right\rceil ,
   L_{12} = \sqrt{\frac{\sum_{h' \in H_{12}} K_{12}(h')}{p}}
\end{align*}
and similarly for the other two types $j=1,2$:
\begin{align*}
  & \forall  h \in H_{j}: \quad  K_{j}(h) = m_j(h) \\
  & p_{h} =  \left\lceil p \frac{K_j(h)}{\sum_{h' \in H_j} K_j(h')}  \right\rceil ,
   L_{j} =  \sqrt{\frac{\sum_{h' \in H_j} K_j(h')}{p}}
\end{align*}
The algorithm operates as follows:

\begin{packed_enum}
\item \label{enum:join:1} Compute $S_1(x,z),S_2(y,z)$ using a
  hash-join on the variable $z$, on all light hitters, using all $p$ servers.
\item \label{enum:join:2} For a heavy hitter $h \in H_{12}$: use $p_{h}$ servers
to compute the cartesian product $S_1(x,h), S_2(y,h)$.
\item \label{enum:join:3} For a heavy hitter $h \in H_{1}$: use $p_{h}$ servers.  
Hash-partition $S_1(x,h)$ on $x$, and broadcast $S_2(y,h)$ to all $p_{h}$ servers.
\item For a heavy hitter $h \in H_{2}$: similar to above.
\end{packed_enum}

While our description has four logical steps, the algorithm consists
of a single communication step.  All input servers know the heavy
hitters, and thus can compute all quantities $p_{h}$.
Then, in parallel, every input server holding some fragment of $S_j$
will examine each tuple $S_j(a,b)$ and determine in which case it is handled.

For each of the steps of the algorithm, the number of servers used may exceed $p$
(since we round up to the closest integer), but we can show that it will
always be $\Theta(p)$.
Additionally, by examining the complexity of each step, we can conclude that,
with high probability, the maximum load per server is at most $O(L \ln(p))$,
\begin{align} \label{eq:max:join}
  L = \max(m_1/p, m_2/p, L_1, L_2, L_{12})
\end{align}

We end our discussion by showing that the algorithm is almost optimal; we
argue next that any algorithm must incur a load
per server of at least $L$. It suffices to prove that,
for each of the quantities under $\max(\ldots)$ in \eqref{eq:max:join}, any algorithm
must have a load of at least that quantity.  We illustrate here the
main idea of the proof for $L_{12}$.  Let
$[n]$ denote the domain; then the size of the join is $\sum_{h
  \in [n]} m_1(h)m_2(h)$.  Denote $w_{i1}(h)$ the number of tuples of
the form $S_1(x,h)$ received by server $i$; $\mathbf{w}_{1}(h) = (w_{11}(h),
\ldots, w_{p1}(h))$ be the $p$-vector representing how the tuples with
$z=h$ in $S_1$ are distributed to the $p$ servers; and $\mathbf{w}_1 = \sum_{h
  \in [n]} \mathbf{w}_1(h)$ the vector showing the total load per server coming
from the table $S_1$.  We denote similarly $w_{i2}(h)$,
$\mathbf{w}_{2}(h)$, $\mathbf{w}_2$ for $S_2$.
Server $i$ can report at most $w_{i1}(h)w_{i2}(h)$ join tuples
with $z=h$.  Since the servers must find all answers,
\begin{align*}
  & \sum_{h \in [n]} m_1(h)m_2(h) 
   \leq  \sum_{i=1}^p \sum_{h \in [n]}w_{i1}(h)w_{i2}(h) \\
  & \leq  \sum_{i=1}^p  \sum_{h \in [n]}  w_{i1}(h) \cdot \sum_{h \in [n]} w_{i2}(h) 
   = \iprod{\mathbf{w}_{1}, \mathbf{w}_{2}} \\
  &  \leq \norm{\mathbf{w}_{1}}_2 \norm{\mathbf{w}_{2}}_2 
    \leq p \norm{\mathbf{w}_1}_\infty
  \norm{\mathbf{w}_2}_\infty \leq p L^2
\end{align*}
In the last line we used the fact that $\norm{\mathbf{w}}_2 \leq \sqrt{p}
\norm{\mathbf{w}}_\infty$ and that 
$L \geq \max (\norm{\mathbf{w}_1}_\infty,\norm{\mathbf{w}_2}_\infty)$.

\subsection{An Algorithm for the General Case}

We now generalize some of the ideas for the simple join to an arbitrary
conjunctive query $q$. 
Extending the notion for simple joins, for each relation $S_j$ with
$|S_j|=m_j$ we say
that a partial assignment $\bh_j$ to a subset $\bx_j\subset vars(S_j)$ is a {\em heavy hitter}
if and only if the number of tuples, $m_j(\bh_j)$, from $S_j$ that contain $\bh_j$ satisfies
$m_j(\bh_j)> m_j/p$.
As before, there are $O(p)$ such heavy hitters.
We will assume that each input server knows the entire set of heavy hitters for 
all relations.  

For simplicity we assume that $p$ is a power of 2.
We will not produce quite as smooth a bound as we did for the simple join, but
we will show that the bound we produce is within a $\log^{O(1)} p$ factor
of the optimal.
To do this, for each relation $S_j$ and subset of variables $\bx_j$, we define
$\log_2 p$ bins for the frequencies, or degrees of each
of the heavy hitters. The $b$-th bin, for $b =1, \dots,  \log_2 p$ will
contain all heavy hitters $\bh_j$ with $m_j/2^{b-1} \ge m_j(\bh_j)> m_j/2^b$.  
The last bin, a bin of light hitters with $b=\log_2 p+1$, will contain all assignments
$\bh_j$ to $\bx_j$ that are not heavy hitters. 
Notice that, when $\bx_j=\emptyset$, then the only non-empty
  bin is the first bin, the only heavy hitter is the empty tuple
  $\bh_j = ()$, and $m_j(\bh_j) = m_j$.

For a bin $b$ on $\bx_j$ 
define $\beta_b=\log_p(2^{b-1})$;
observe that for each heavy hitter bin, there are at
most $2p^{\beta_b}$ heavy hitters in this bin, and for the last bin we have
$\beta_b= 1$. Instead of identifying each bin using its index $b$, we
identify each bin by $\beta_b$, called its {\em bin exponent}, along
with the index of the relation $S_j$ for which it is defined, and the set
$\bx_j \subset vars(S_j)$. Note that $0 = \beta_1 < \beta_2 <
  \cdots < \beta_{\log_2 p+1}=1$.

\newcommand{\B}{\ensuremath{\mathcal B}}

\begin{definition}[Bin Combination]
Let $\bx \subset V = vars(q)$, and define $\bx_j = \bx \cap vars(S_j)$.
A pair $\B=(H,(\beta_j)_j)$ is called a {\em bin combination} if
(1) $\beta_j=0$ for every $j$ where $\bx_j= \emptyset$, and
(2) there is some consistent assignment $\bh$ to $\bx$ such
that for each $j$ with $\bx_j\ne \emptyset$ the induced assignment
$\bh_j$ to $\bx_j$ has bin exponent $\beta_j$ in relation $S_j$.
We write $C(\B)$ for the set of all such assignments $\bh$.
\end{definition}

%
%

\newcommand{\N}{N_{{\mathrm bc}}}

Our algorithm allocates $p$ virtual processors to each 
bin combination and handles associated inputs separately.
There are $O(\log p)$ bin choices for each relation
and therefore at most $\log^{O(1)} p$ bin combinations in total, so
at most $p\log^{O(1)}p$ virtual processors in total.
Let $\N$ be the number of possible bin combinations.
As in the join algorithm (\autoref{subsec:join}), within each
  bin combination we partition the $p$ servers among the heavy
  hitters, using $p_\bh=p^{1-\alpha}$ servers for heavy hitter $\bh$
  (independent of $\bh$, since we have ensured  complete
  uniformity within a bin combination). However, we can only process
  $p^\alpha \leq p$ heavy hitters in every bin combination: in
  general, we may have $C(\B)>p$, e.g. when $\bx$ contains a variable
  $x_1$ in $S_1$ and a variable $x_2$ in $S_2$, then there may be up
  to $p \times p$ heavy hitters in this bin combination.

For a bin combination $\B$ and assignment $\bh\in C(\B)$, 
write $m_j(\bh)$ for $m_j(\bh_j)$.
For each $\B$ we will define below a set $C'(\B)\subseteq C(\B)$ with
$|C'(\B)|\le p$ and sets $S^{(\B)}_j\subseteq S_j$ of tuples for 
$j\in [\ell]$ that extend  $\bh_j$ for some $\bh\in C'(\B)$.

\begin{sloppypar}
The algorithm for $\B$ will 
compute all query answers that are joins of $(S^{(\B)}_j)_j$, by executing the
HC algorithm for a particular choice of exponents
given. 
The share exponents for the HC algorithm will be provided
by a modification of the vertex 
covering primal LP~\eqref{eq:primal:lp}, which describes an algorithm that
suffices for all light hitters.
Recall that in this LP, $\mu_j=\log_p M_j$ and $\lambda$ is $\log_p L$ for the
load $L$.
%
%
That LP corresponds to the bin combination $\B_\emptyset$ which has
$\bx=\emptyset$ and all $\beta_j=0$.
In this case $C'(\B_\emptyset)=C(\B_\emptyset)$ has 1 element, the empty
partial assignment.
More generally, write $\alpha=\log_p |C'(\B)|$, 
the LP associated with our algorithm for bin combination $\B$ is:
\begin{align}
\text{mininimize} & \quad  \quad \lambda\label{B}\\
 \mbox{ subject to} & \nonumber \\
\forall j \in [\ell]: & \quad \lambda+\sum_{x_i \in vars(S_j)-\bx_j} e_i \ge \mu_j-\beta_j \nonumber\\
& \sum_{i \in V-\bx} e_i\le 1-\alpha\nonumber\\
\forall i \in V-\bx: & \quad e_i \ge 0, \quad \lambda \geq 0 \nonumber
\end{align}
Thus far, this is only fully specified for $\B=\B_\emptyset$ since we
have not defined $C'(\B)$ for other choices of $\B$.  We define $C'(\B)$
inductively based on optimal solutions
$(\lambda^{(\B')},(e^{(\B')}_i)_{i\in V-\bx'})$ to the above LP applied to bin
combinations $\B'$ with $\bx'\subset \bx$ as follows.  (Such solutions may
not be unique but we fix one arbitrarily for each bin combination.)
\end{sloppypar}

For $\bh'\in C'(\B')$, we say that a heavy hitter $\bh_j$ of $S_j$ that is 
an extension of $\bh'_j$ to $\bx_j$  is
{\em overweight} for $\B'$ if there are more than
$\N \cdot m_j/p^{\beta_j+\sum_{i\in \bx_j-\bx'_j} e^{(\B')}_i}$ elements of
$S_j$ consistent with $\bh_j$.
$C'(\B)$ consists of all assignments $\bh\in C(\B)$ such that there is 
some $j\in [\ell]$, some bin combination $\B'$ on set $\bx'\subset \bx$
such that $\bx-\bx'\subseteq vars(S_j)$, and some $\bh'\in C'(\B')$ such that
$\bh$ is an extension of $\bh'$ and $\bh_j$ is an overweight heavy hitter
of $S_j$ for $\B'$.
The following lemma, which is proved in the appendix, shows that $\alpha\le 1$ and thus
that the LP for $\B$ is feasible.

\begin{lemma}
\label{lem:atmostp}
For all bin combinations $\B$, $|C'(\B)|\le p$.
\end{lemma}

Let $A_\B\subseteq [\ell]$ be the set of all $j$ such that $\bx_j\ne \emptyset$.
For each $j\in [\ell]-A_\B$, let $S^{(\B)}_j$ consist of all tuples in $S_j$
that
do not contain any heavy hitter $\bh''_j$ of $S_j$ that is overweight for $\B$.
For each $j\in A_\B$, and $\bh\in C'(\B)$ let $S^{(\B)}_j(\bh)$
consist of all tuples in $S_j$
that contain $\bh_j$ on $\bx_j$ (with bin exponent $\beta_j$)
but do not contain any heavy hitter $\bh''_j$ of $S_j$ that is overweight for $\B$ and a proper extension of $\bh_j$.
$S^{(\B)}_j$ will be the union of all $S^{(\B)}_j(\bh)$ for all $\bh\in C'(\B)$.

For each of the $p^{\alpha}$ heavy hitter assignments $\bh\in C'(\B)$ 
the algorithm uses $p^{1-\alpha}$ virtual processors to compute the 
join of the subinstances $S^{(\B)}_j(\bh)$ for $j\in [\ell]$.
Those processors will be allocated using the HC algorithm
that assigns $p^{e^{(\B)}_i}$ shares to each
variable $x_i \in V-\bx$.  It remains to show that the load of the
algorithm for $\B$ is within a $\log^{O(1)}p$ factor of $p^{\lambda^{(\B)}}$,
where $\lambda^{(\B)}$ is given by the LP for $\B$.



\begin{lemma}
\label{hashapplication}
Let $\bh$ be an assignment to $\bx$ that is consistent with bin combination $\B$.
If we hash each residual relation $S^{(\B)}_j(\bh)$ on $vars(S_j)-\bx_j$ 
using $p^{e^{(\B)}_i}$ values for each $x_i\in vars(S_j)-\bx_j$, each processor receives
$$O \left( (\N\cdot \ln p)^{r'} \cdot m_j/p^{\min(\beta_j+\sum_{i\in vars(S_j)-\bx_j} e^{(\B)}_i,1)} \right)$$
values with high probability, where $r'=\max_j (r_j-|\bx_j|)$.
\end{lemma}


\begin{proof}
For $j\in [\ell]-A_\B$, $S^{(\B)}_j$ only contains tuples of $S_j$ that are not
overweight for $\B$, which means that for every $\bx''_j\subseteq vars(S_j)$
and every heavy hitter assignment $\bh''$ to the variables of $\bx''_j$, there
are at most
$$\N \cdot m_j/p^{\beta_j+\sum_{i\in \bx''_j} e^{(\B)}_i}=
\N \cdot m_j/p^{\beta_j+\sum_{i\in \bx''_j-\bx_j} e^{(\B)}_i}$$
elements of $S_j$ consistent with $\bh''$.   
Every other assignment $\bh''$ to the variables of $\bx''_j$ is a light hitter
and therefore is contained in at most $m_j/p$ consistent tuples of $S_j$.
For $j\in A_\B$, we obtain the same bound, where the only difference is that
we need to restrict things to extensions of $\bh_j$.
This bound gives the smoothness condition on $S^{(\B)}_j(h)$ necessary to apply
\autoref{lemma:hashing} to each relation $S^{(\B)}_j(\bh)$ and yields the
claimed result. 
\end{proof}

As a corollary, we obtain:

\begin{corollary}
\label{perBload}
Let $L_{min}=\max_j (m_j/p)$.
The maximum load of our algorithm using $p$ (virtual) processors for $\B$ is 
$O((\N\cdot \ln p)^{r_{\max}}\cdot  \max(L_{min}, p^{\lambda} ))$
with high probability, where $\lambda=\lambda^{(\B)}$ is the optimum of the LP
for $\B$ and $r_{\max}$ is the maximum arity of any $S_j$.
\end{corollary}

\begin{proof}
There are $p^{1-\alpha}$ processors allocated to each $\bh$ and $p^\alpha$ 
such assignments $\bh$ so that the
maximum load per $\bh$ is also the maximum overall load.  
Given the presence of the $L_{min}$ term,
it suffices to show that the maximum load per $\bh$ 
due to relation $S^{(\B)}_j(\bh)$ is at most
$(\ln p)^{k-|\bx|}\cdot \max(m_j/p,p^{\lambda})$ for each $j\in [\ell]$.
Observe that by construction, $p^\lambda$ is the smallest value such that
$p^{\lambda} \cdot p^{\beta_j+\sum_{x_i\in vars(S_j)-\bx_j} e^{(\B)}_i}\ge m_j$
for all $j$ and $\sum_{i\in \bx} e^{(\B)}_i\le 1-\alpha$.  
Lemma~\ref{hashapplication} then implies that the load due to relation $S^{(\B)}_j(\bh)$
is at most a polylogarithmic factor times 
$$\max(m_j/p,m_j/p^{\beta_j+\sum_{x_i\in vars(S_j)-H_j} e^{(\B)}_i})$$ 
which is at most $\max(m_j/p,p^{\lambda})$.
\end{proof}


\begin{lemma}
\label{lem:bincombocover}
Every tuple in the join of $(S_j)_{j\in [\ell]}$ is contained in a join of
subrelations $(S^{(\B)}_j)_{j\in [\ell]}$ for some bin combination $\B$.
\end{lemma}

\newcommand{\bt}{{\mathbf t}}

\begin{proof}
Observe first that every join tuple is consistent with the empty bin
combination $\B_\emptyset$.  Therefore the join of
$(S^{(\B_\emptyset)}_j)_{j\in [\ell]}$ contains all join tuples that do not
contain an overweight heavy hitter $\bh_j$ for any relation $S_j$ with
respect to $\B_\emptyset$ (and therefore contains all join tuples that
are not consistent with any heavy hitter). 
Now fix a join tuple $\bt$ that is overweight for $\B_\emptyset$. 
By definition, there is an 
associated relation $S_{j_1}$ and $\bx^1\subset vars(S_{j_1})$ such that
$\bh^{1}=(\bt_{\bx^1})$ is an overweight heavy hitter of $S_{j_1}$ for
$\B_\emptyset$.   Let $\B_1$ be
the bin combination associated with $\bh^{1}$.
By definition $\bh^{1}\in C'(\B_1)$.  
Now either $\bt$ is contained in the join of
$(S^{(\B_1)}_j)_{j\in [\ell]}$ and we are done or there is some relation
$S_{j_2}$ and 
$\bx^2$ such that $\bx^2-\bx^1\subset vars(S_{j_2})$ such that
$\bh^{2}=(\bt_{\bx^2})$ has the property that $\bh^{2}_{j_2}$ is
an overweight heavy hitter of $S_{j_2}$ for $\B_1$.   Again, in the latter case,
if $\B_2$ is the bin combination associated with $\bh^{2}$ then 
$\bh^2\in C'(\B_2)$ by definition and we can repeat the previous argument
for $\B_2$ instead of $\B_1$.
Since the number of variables grows at each iteration, we can repeat
this at most $k$ times before finding a first $\B_r$ such that $\bt$ is not
associated with any overweight heavy hitter for $\B_r$.  In this case
$\bt$ will be computed in the join of
$(S^{(\B_r)}_j)_{j\in [\ell]}$.
\end{proof}


We can now prove the main theorem.

\begin{theorem}
\label{th:binalgorithm}
The algorithm that computes the joins of $(S_j^{(\B)})_{j\in [\ell]}$ for every
bin combination $\B$ using the HC algorithm as described above
has maximum load  $L \leq \log^{O(1)}p \cdot \max_{\B} p^{\lambda^{(\B)}}$.
\end{theorem}

\begin{proof}
There are only $\log^{O(1)}p$ choices of $\B$ and for each choice of $\B$,
by \autoref{perBload},
the load with $p$ virtual processors
is $\log^{O(1)} p\cdot \max(L_{min},p^{\lambda^{(\B)}})$.
To derive our claim it suffices to show that we can remove the $L_{min}$ term.
Observe that in the original LP which corresponds to an empty bin combination
$\B$, we have 
$\lambda+\sum_{i \in S_j} e_i\ge \mu_j$  for each $j\in [\ell]$
and $\sum_i e_i\le 1$.   This implies that $\lambda\ge \mu_j-1$ and hence
$p^{\lambda}\ge m_j/p$ for each $j$, so $p^\lambda\ge L_{min}$.
\end{proof}

To show the optimality of the algorithm, we apply \autoref{th:lower:skew}
(which we prove
in the next subsection) for a particular bin combination $\B$ with set $\bx$.
Notice that $m_j(\bh_j) \geq m_j/(2p^{\beta_j})$ and further the
number of $\bh$ is at most $p^\alpha$. Hence the lower bound we
obtain from \autoref{eq:lx} is of the form $L \geq (\prod_j (m_j/p^{\beta_j})^{u_j}/p^{1-\alpha})^{1/u}/2$ for any packing $\bu$
of $q_{\bx}$ that saturates $\bx$, which can be easily seen to be $p^{\lambda^{(\B)}}/2$, using a duality  argument with the LP in \eqref{B}.

%% file: section4-lower-bound.tex
\subsection{ Lower Bound}
\label{sec:lower-bound}

In this section we give a lower bound for the load of any deterministic
algorithm that computes a query $q$, and generalizes the lower bound
in \autoref{th:lower:uniform}, which was over databases
with cardinality statistics $\bM$.  Our new lower bound generalizes this to
databases with a fixed degree sequence: if the degree sequence is
skewed, then the new bounds can be stronger, proving that skew in
the input data makes query evaluation harder.

Let $\bx_j \subseteq vars(S_j)$ and $d_j = |\bx_j|$.  A {\em
  statistics of type $\bx_j$} is a function $m_j : [n]^{d_j}
\rightarrow \mathbb{N}$.  An instance $S_j$ satisfies the statistics
$m_j$ if for any tuple $\bh_j \in [n]^{d_j}$, its frequency is
precisely $m_j(\bh_j)$, in other words $|\sigma_{\bx_j = \bh_j}(S_j)|
= m_j(\bh_j)$.  For an example, if $S(x,y)$ is a binary relation, then
an $x$-statistics is a degree sequence $m(1), m(2), \ldots, m(n)$; for
another example, if $\bx_j = \emptyset$ then an $\bx$-statistics
consists of a single number, which denotes the cardinality of $S_j$.
In general, the $\bx_j$ statistics define uniquely the cardinality of
$S_j$, as $|S_j| = \sum_{\bh_j \in [n]^{d_j}} m_j(\bh_j)$.

Fix a set of variables $\bx$ from $q$, let $d = |\bx|$, and denote
$\bx_j = \bx \cap vars(S_j)$ for every $j$.  A {\em
  statistics of type $\bx$ for the database} is a vector $\bm = (m_1,
\ldots, m_\ell)$, where each $m_j$ is an $\bx_j$-statistics for $S_j$.
We associate with $m$ the function $m : [n]^k \rightarrow
(\mathbb{N})^\ell$, $m(\bh) = (m_1(\bh_1), \ldots, m_\ell(\bh_\ell))$;
here and in the rest of the section, $\bh_j$ denotes
$\pi_{\bx_j}(\bh)$, \ie the restriction of $\bh$ to the variables in
$\bx_j$.  When $\bx = \emptyset$, then $m$ simply consists of $\ell$
numbers, each representing the cardinality of a relation; thus, a
$\bx$-statistics generalizes the cardinality statistics from
\autoref{sec:basic}.  Recall that we use upper case $\bM = (M_1, \ldots,
M_\ell)$ to denote the same statistics expressed in bits,
i.e. $M_j(\bh_j) = a_j m_j(\bh_j) \log n$.  As before, we fix some
constant $0 < \delta < 1$, and assume every relation
$S_j$, has cardinality $\leq n^{\delta}$.

In this section, we fix statistics $\bM$ of type $\bx$ and consider the
probability space where the instance is chosen uniformly at
random over all instances that satisfy $\bM$.

To prove the lower bound we need some notations.  Let $q_\bx$ be the
{\em residual query}, obtained by removing all variables $\bx$, and
decreasing the arities of $S_j$ as necessary (the new arity of $S_j$
is $a_j - d_j$).  Clearly, every fractional edge packing of $q$ is
also a fractional edge packing of $q_\bx$, but the converse does not
hold in general.  Let $\bu$ be a fractional edge packing of $q_\bx$.
We say that $\bu$ {\em saturates} a variable $x_i \in \bx$, if
$\sum_{j: x_i \in vars(S_j)} u_j \geq 1$; we say that $\bu$ saturates $\bx$ if
it saturates all variables in $\bx$.  For every fractional edge
packing $\bu$ of $q_\bx$ that saturates $\bx$, denote $u = \sum_{j=1}^{\ell} u_j$
and, using $K$  defined in \eqref{eq:kum}:
\begin{align}
  L_\bx(\bu, \bM, p) = & \left( \frac{\sum_{\bh \in [n]^{d}} K(\bu, \bM(\bh))}{p} \right)^{1/u}\label{eq:lx}
\end{align}
Further denote $L_{\texttt{lower}} = \max_{\bu} L_\bx(\bu, \bM,p)$, and
let $c$ be the constant $c = \min_j \frac{a_j-d_j-\delta}{3 a_j}$.

\begin{theorem} 
  \label{th:lower:skew} 
  Fix statistics $\bM$ of type $\bx$, and consider any
  deterministic MPC algorithm that runs in one communication round on
  $p$ servers and has maximum load $L$.  Then, for any edge packing
  $\bu$ of $q$ that saturates $\bx$, 
%
   any algorithm that computes $q$ correctly 
   must have load $L \geq c L_\bx(\bu,\bM,p)$ bits.
\end{theorem}

%
%

%

Note that, when $\bx = \emptyset$ then $L_\bx(\bu, \bM,p) = L(\bu,\bM,p)$,
defined in \eqref{eq:lump}; therefore, our theorem is a
generalization of the simpler lower bound \autoref{th:lower:uniform}.
Before we prove the theorem, we show an example.

\begin{example}
  Consider $q(x,y,z) = S_1(x,z),S_2(y,z)$.  If
  $\bx = \emptyset$ then the lower bound is $\max_\bu
  L(\bu,(M_1, M_2),p)$, which is $\max(M_1/p, M_2/p)$, because there are two
  packings in $pk(q)$: $(1,0)$ and $(0,1)$.  For $\bx =
  \set{z}$, the residual query is $q_\bx = S_1(x),S_2(y)$ and its
  sole packing is $(1,1)$, which saturates the variable $z$.
  This gives us a new lower bound: $\sqrt{\sum_{h \in [n]} M_1(h) M_2(h)/p}$.  
  The lower bound is the maximum of these
  quantities, and we have given in \autoref{sec:algo} an algorithm
  that matches this bound.

  Next, consider $C_3$.  In addition to the lower bounds in \autoref{ex:triangle}, we obtain new
  bounds by setting $\bx = \set{x_1}$. The residual query is 
  $S_1(x_2), S_2(x_2,x_3), S_3(x_3)$ and $(1,0,1)$  is a packing that
  saturates $x_1$ (while for example $(0,1,0)$ does not).
  This gives us a new lower bound, of the form $\sqrt{\sum_{h \in [n]} m_1(h)m_3(h)/p}$.
\end{example}

In the rest of the section we prove \autoref{th:lower:skew}.

Let $J(\bx) = \setof{j}{\bx \cap vars(S_j) \neq \emptyset}$ denote the
set of relations that contain some variable from $\bx$.
Fix a concrete instance $S_j$, and let $\ba_j \in S_j$.  We write
$\ba_{j} | \bh$ to denote that the tuple $\ba_j$ from $S_j$ matches
with $\bh$ at their common variables, and denote $(S_j)_\bh$ the
subset of tuples $\ba_j$ that match $\bh$: $(S_j)_\bh
=\setof{\ba_j}{\ba_j \in S_j, \ba_j | \bh}$.  Let $I_{\bh}$ denote the
restriction of $I$ to $\bh$, in other words $I_\bh = ((S_1)_\bh,
\ldots, (S_\ell)_\bh)$.

When $I$ is chosen at random over the probability space defined by the
$\bx$-statistics $\bM$, then, for a fixed tuple $\bh \in [n]^d$, the
restriction $I_\bh$ is a uniformly chosen instance over all instances
with cardinalities $\bM(\bh)$, which is precisely the probability space
that we used in the proof of \autoref{th:lump}.  In particular, for
every $\ba_j \in [n]^{d_j}$, the probability that $S_j$ contains
$\ba_j$ is $P(\ba_j \in S_j) = m_j(\bh_j) / n^{a_j-d_j}$; thus, our
proof below is an extension of that of \autoref{th:lump}.

\autoref{lem:expected_size} immediately gives:
\begin{lemma}
$ \E[|q(I_{\bh})|] =   n^{k-d} \prod_{j=1}^{\ell} \frac{m_j(\bh_j)}{ n^{a_j-d_j}}$
\end{lemma}

Let us fix some server and let $m(I)$ be the message the server receives on input $I$.
For any fixed value $m$ of $m(I)$, let $K_m(S_j)$ be the set of tuples from relation $S_j$
{\em known} by the server. Let $w_j({\ba_j})$ to denote the probability
that the server knows the tuple $\ba_j \in S_j$. In other words
$w_j(\ba_j) = P(\ba_j \in K_{m_j(S_j)}(S_j))$, where the probability
is over the random choices of $S_j$.  This is obviously upper bounded
by $P(\ba_j \in S_j)$:
\begin{align}
  w_j(\ba_j|\bh) \leq m_j(\bh_j) / n^{a_j-d_j} \label{eq:w:n}
\end{align}

We derive a second upper bound by exploiting the fact that the server
receives a limited number of bits:
\begin{lemma}
\label{lem:L-bound}
Let $L$ be the number of bits a server receives. Then, 
$\sum_{\ba_j \in [n]^{a_j}} w_j(\ba_j) \leq \frac{L}{c a_j \log(n)}$.
\end{lemma}

\begin{proof}
Since $\sum_{\ba_j \in [n]^{a_j}} w_j(\ba_j)  = \E[|K_{m(S_j)}(S_j)|]$, we will bound
the right hand side. Now, notice that:
  \begin{align} \label{eq:entropy:one}
   H(S_j) &=H(m(S_j))+ \sum_{m}P(m(S_j)=m)\cdot H(S_j|m(S_j)=m)\nonumber\\
  &\le L +  \sum_{m}P(m(S_j)=m)\cdot H(S_j|m(S_j)=m)
\end{align}

For every $\bh$, let $K_m((S_j)_{\bh})$ denote the known tuples from the restriction of 
$S_j$ to $\bh$. Following the proof of \autoref{lem:entropy_ratio},
for the constant $c = \log e +1$, 
\begin{align*}
 H(S_j|m(S_j)=m)  & \leq \sum_{\bh} \left(1-\frac{|K_m((S_j)_{\bh})|}{c m_j(\bh_j)} \right) \log \binom{n^{a_j-d_j}}{m_j(\bh_j)} \\
 & = H(S_j) - \sum_{\bh} \frac{|K_m((S_j)_{\bh})|}{c m_j(\bh_j)} \log \binom{n^{a_j-d_j}}{m_j(\bh_j)} \\
 & \leq H(S_j) - \sum_{\bh} \frac{|K_m((S_j)_{\bh})|}{c m_j(\bh_j)} m_j(\bh_j)(a_j-d_j-\delta) \log (n)  \\ 
 & = H(S_j) - (1/c) \cdot |K_m(S_j)| (a_j-d_j-\delta) \log (n) 
\end{align*}
Plugging this in \autoref{eq:entropy:one}, we have:
$$ H(S_j) \leq L + H(S_j) - (1/c) \cdot \E[|K_m(S_j)|] (a_j-d_j-\delta) \log (n) $$
or equivalently, since $c \leq 3$:
\begin{align*}
\E[|K_m(S_j)|]  \leq \frac{3 L}{(a_j-d_j-\delta) \log(n)}
\end{align*}
This concludes our proof.
\end{proof}

Let $q_{\bx}$ be the residual query, and recall that $\bu$ is a
fractional edge packing that saturates $\bx$.  Define the {\em
  extended query} ${q_{\bx}}'$ to consists of $q_\bx$, where we add a
new atom $S_i'(x_i)$ for every variable $x_i \in
vars(q_\bx)$.  Define $u'_i = 1 - \sum_{j: i \in S_j} u_j$.
In other words, $u'_i$ is defined to be the slack at the variable
$x_i$ of the packing $\bu$.  The new edge packing
$(\mathbf{u},\mathbf{u}')$ for the extended query $q_\bx'$ has no more
slack, hence it is both a tight fractional edge packing and a tight
fractional edge cover for $q_{\bx}$.  By adding all equalities of the
tight packing we obtain:
$$\sum_{j=1}^{\ell} (a_j-d_j) u_j + \sum_{i=1}^{k-d} u'_i = k-d$$

We next compute how many output tuples from $q(I_{\bh})$ will be known
in expectation by the server. Note that $q(I_{\bh}) = q_{\bx}(I_{\bh})$, and thus:
\begin{align*}
\allowdisplaybreaks
\E[|K_m(q(I_{\bh}))|] & = \E[|K_m(q_{\bx}(I_{\bh}))|] 
 =  \sum_{\ba \in [n]^{k-d}} \prod_{j=1}^{\ell} w_j({\ba_j|\bh}) \\
& =  \sum_{\ba \in [n]^{k-d}} \prod_{j=1}^{\ell} w_j({\ba_{j}|\bh}) 
         \prod_{i=1}^{k-d} w_i'({\ba_i})\\      
 & \leq   \prod_{i=1}^{k-d} n^{u_i'} \cdot 
  \prod_{j =1}^{\ell} \left( \sum_{\ba \in [n]^{a_j-d_j}} w_j({\ba}|\bh)^{1/u_j} \right)^{u_j} 
\end{align*}
By writing $w_j(\ba | \bh)^{1/u_j} = w_j({\ba} | \bh)^{1/u_j-1} w_j({\ba} | \bh)$, we can bound the quantity as follows:
\begin{align*}
 \sum_{\ba \in [n]^{a_j-d_j}} w_j({\ba | \bh})^{1/u_j} 
 & \leq \left( \frac{m_j(\bh_j)}{ n^{a_j-d_j}} \right)^{1/u_j-1} \sum_{\ba \in [n]^{a_j-d_j}} w_j({\ba}|\bh) \\
\quad & = (m_j(\bh_j)   n^{d_j-a_j})^{1/u_j-1}  L_j(\bh) \\
\end{align*}
where $L_j(\bh) = \sum_{\ba \in [n]^{a_j-d_j}} w_j({\ba}|\bh)$. Notice that for every relation $S_j$,
$\sum_{\bh_j \in [n]^{d_j}} L_j(\bh_j)  = \sum_{\ba_j \in [n]^{a_j}} w_j(\ba_j)$.
\begin{align}
 \mathbf{E}[|K_m(q(I_{\bh}))|] 
& \leq  n^{\sum_{i=1}^{k-d} u_i'}  \prod_{j=1}^{\ell} \left( L_{j}(\bh)   m_j(\bh_j)^{1/u_j-1}  n^{(d_j-a_j)(1/u_j-1)} \right)^{u_j}  \nonumber\\
%
%
& =  \prod_{j=1}^{\ell} L_{j}(\bh)^{u_j} \cdot \prod_{j=1}^{\ell} m_j(\bh_j)^{-u_j} \cdot \E [|q(I_{\bh})|] \label{eq:lastline}
\end{align}

Summing over all $p$ servers, we obtain that the expected number of answers that can be output for $q(I_{\bh})$ is at most $p \E[|K_m(q(I_{\bh}))|]$. If some $\bh \in [n]^d$ this number is not at least $\E[|q(I_{\bh})|]$, the algorithm will fail to compute $q(I)$. Consequently, for every $\bh$ we must have that 
$ \prod_{j=1}^{\ell} L_{j}(\bh_j)^{u_j} 
\geq  (1/p) \cdot \prod_{j =1}^{\ell} m_j(\bh_j)^{u_j}$.
Summing the inequalities for every $\bh \in [n]^{d}$:
$$ (1/p) \cdot \sum_{\bh} \prod_{j=1}^{\ell} m_j(\bh_j)^{u_j}
\leq \sum_{\bh} \prod_{j=1}^{\ell} L_j(\bh_j)^{u_j}$$

To finish the lower bound, we will apply Friedgut's inequality on the right hand side of the above inequality. 

\begin{lemma}
If $\mathbf{u}$ saturates $\bx$ in query $q$,
$$ \sum_{\bh} \prod_{j=1}^{\ell} L_j(\bh_j)^{u_j} \leq 
\prod_{j=1}^{\ell} \left( \sum_{\bh_j} L_j(\bh_j) \right)^{u_j}$$
\end{lemma}

Since we have 
$M_j = a_j m_j \log(n)$, and $\sum_{\bh_j} L_j(\bh_j) \leq L/(c a_j \log(n))$ 
(from \autoref{lem:L-bound}), we obtain that for any edge packing $\bu$ that 
saturates $\bx$, it must be that
\begin{align*}
 L \geq  c \cdot \left( \frac{\sum_{\bh \in [n]^d} \prod_{j} M_j(\bh_j)^{u_j}}{p} \right)^{1/u}  
\end{align*}

%% file: mapreduce.tex
\section{Map-Reduce Models}
\label{sec:map-reduce}

In this section, we discuss the connection between the MPC model and the Map-Reduce model presented by Afrati et. al~\cite{ASSU13}. 
In contrast to the MPC model, where the number of servers $p$ is the main parameter, in the model of~\cite{ASSU13} the main parameter is an upper bound $q$ on the number of input tuples a reducer can receive, which is called {\em reducer size}. Given an input $I$, a Map-Reduce algorithm is restricted to deterministically send each input tuple independently to some reducer, which will then produce all the outputs that can be extracted from the received tuples. If $q_i \leq q$ is the number of inputs assigned to the $i$-th reducer, where $i=1, \dots, p$, we define the {\em replication rate} $r$ of the algorithm $ r = \sum_{i=1}^p q_i / |I|$.

In~\cite{ASSU13}, the authors provide lower and upper bounds on $r$ with respect to $q$ and the size of the input. However, their results are restricted to binary relations where all sizes are equal, and they provide matching upper and lower bounds only for a subclass of such queries. We show next how to apply the results of this paper to remove these restrictions, and further consider an even stronger computation model for our lower bounds.

First, we express the bound on the data received by the reducers in {\em bits}: let $L$ be the maximum number of bits each reducer can receive. The input size $|I|$ is also expressed in bits. We will also allow any algorithm to use randomization. Finally, we relax the assumption on how the inputs are communicated to the reducers: instead of restricting each input tuple to be sent independently, we assume that the input data is initially partitioned into a number of {\em input servers} $p_0$ (where $p_0$ must be bigger than the query size), and allow the algorithm to communicate bits to reducers by accessing all the data in one such input server. Notice that this setting allows for stronger algorithms that use input statistics to improve communication. 

If each reducer receives $L_i$ bits, the replication rate $r$ is defined as $r = \sum_{i=1}^p L_i / |I|$. Notice that any algorithm with replication rate $r$ must use $p \geq (r |I|) / L$ reducers. Now, let $q$ be a conjunctive query, where $S_j$ has $M_j = a_j m_j \log(n)$ bits ($n$ is the size of the domain). 

\begin{theorem}
\label{th:map-reduce}
Let $q$ be a conjunctive query where $S_j$ has size (in bits) $M_j$. Any algorithm that computes $q$ with reducer size $L$, where $L \leq M_j$ for every $S_j$\footnote{if $L > M_j$, we can send the whole relation to any reducer without cost} must have replication rate
$$ r \geq \frac{c^{u} L}{\sum_j M_j} \max_{\bu} \prod_{j=1}^{\ell} \left( \frac{M_j}{L} \right)^{u_j} $$ 
where $\bu$ ranges over all fractional edge packings of $q$.
\end{theorem}

\begin{proof}
  Let $f_i$ be the fraction of answers returned by server $i$, in
  expectation, where $I$ is a randomly chosen database with statistics
  $\bM$.  By  \autoref{th:lower:uniform},
   $ f_i \leq  \frac{L_i^u}{c^u K(\bu, \bM)}$.
  Since we assume all answers are returned,
  \begin{align*}
    1 \leq & \sum_{i=1}^p f_i = \sum_{i=1}^p \frac{L_i^u}{c^u K(\bu, \bM)}
      =  \frac{\sum_{i=1}^p L_i L_i^{u-1}}{c^u K(\bu, \bM)} 
      \leq  \frac{L^{u-1} \sum_{i=1}^p L_i}{c^u K(\bu, \bM)}
      =   \frac{L^{u-1} r |I|}{c^u K(\bu, \bM)}
  \end{align*}
  where we used the fact that $u \geq 1$ for the optimal $\bu$. The claim follows by using the definition of $K$
  \eqref{eq:kum} and noting that $|I| = \sum_j M_j$.
\end{proof}

%

We should note here that, following from our analysis in~\autoref{sec:basic}, the bound provided by \autoref{th:map-reduce} is matched by the \textsc{HyperCube} algorithm with appropriate shares. We illustrate \autoref{th:map-reduce} with an example. 

\begin{example}[Triangles] Consider the triangle query $C_3$ and assume that all sizes are equal to $M$. 
In this case, the edge packing that maximizes the lower bound is the one that maximizes $\sum_j u_j$,
$(1/2, 1/2, 1/2)$. Thus, we obtain a bound $\Omega(\sqrt{M/L})$ for the replication rate. This is exactly the formula
proved in~\cite{ASSU13}, but notice that we can derive bounds even if the sizes are not equal.
Further, observe that any algorithm must use at least $\Omega((M/L)^{3/2})$ reducers.
\end{example}

%
%

%% file: section5-conclusions.tex
\section{Conclusions}

\label{sec:conclusions}

In this paper we have studied the parallel query evaluation problem on
databases with known cardinalities, and with known cardianlities and
heavy hitters.  In the first case we have given matching lower and
upper bounds that are described in terms of {\em fractional edge
  packings} of the query.  In the second case we have also described
matching lower and upper bounds described in terms of fractional
packings of the residual queries, one residual query for each type of
heavy hitter.

%% file: appendix.tex
\newpage

\appendix

\input{proof-lower-bound}
\input{hashing}

%
%
%

\section{Proofs}

For convenience of notation, for each relation put all tuples in a bin
associated with the empty set of variables and give it bin exponent $0$.
Define an ordering on bin combinations $\B'=(\bx',(\beta'_j)_j)$ and
$\B=(\bx,(\beta_j)_j)$ by $\B'<\B$ iff 
$\bx'\subset \bx$ and for all $j$, $\beta'_j\le \beta_j$.
Observe that for any tuple $\bh\in C(\B)$ and every $\bx'\subset \bx$ there
is a unique bin combination $\B'$ on $\bx'$ with $\B'<\B$ such
that the projection of $\bh$ on $\bx'$, $\bh'$ is in $C(\B')$.

\begin{lemma}
\label{lem:overweight}
Let $\B'=(\bx',(\beta'_j)_j)<\B=(\bx,(\beta_j)_j)$ be bin combinations
and assume that $\log_p C'(\B')=\alpha\le 1$.
If $\bx-\bx'=\bx_j-\bx'_j$ for some $j\in [\ell]$ (that is $\bx$ only adds new
variables from relation $S_j$ to $\bx'$)
then for every $\bh'\in C'(\B')$ there are at most $p^{1-\alpha}/\N$
assignments $\bh''$ such that $\bh_j=\bh'_j\bh''\in C'(\B)$ is an overweight 
heavy hitter of $S_j$ for $\B'$.
\end{lemma}

\begin{proof}
By definition $\bh'_j$ is contained in at most $m_j/p^{\beta'_j}$ tuples
of $S_j$.
However, each of the assignments $\bh_j=\bh'_h\bh''$ that is overweight
for $\B'$ is contained in at least
$\N\cdot m_j/(p^{\beta'_j+\sum_{i\in \bx_j-\bx'_j} e^{(\B')}_i})$ tuples of
$S_j$
because it is overweight.   Therefore the number of different $\bh''$ is
at most $p^{\sum_{i\in \bx_j-\bx'_j} e^{(\B')}_i}/\N$ which is
$\le p^{1-\alpha}/N$ by the properties of the LP for $\B'$.
\end{proof}

We can now prove \autoref{lem:atmostp}

\begin{proof}[Proof of \autoref{lem:atmostp}]
The proof is by induction on the partial order of bin combinations.
By definition $|C'(\B_\emptyset)|=|C(\B_\emptyset)|=1$.
Now consider some arbitrary bin combination $\B=(\bx,(\beta_j)_j)$ 
and assume by induction that $C'(\B')\le p$ for all bin combinations
$\B'=(\bx',(\beta'_j)_j)<\B$.
By definition, for every $\bh\in C(\B)$ that is in $C'(\B)$ there must
be some bin combination $\B'=(\bx',(\beta'_j)_j)<\B$, such that the $\bh'$, the
restriction of $\bh$ to $\bx'$, is in $C'(\B')$ and there is some $j$ such that
$\bx-\bx'=\bx_j-\bx'_j\subset vars(S_j)$ and
$\bh_j$ is an overweight heavy hitter of $S_j$ for bin combination $\B'$.
Let $\alpha=\log_p |C'(\B')|$.
Applying Lemma~\ref{lem:overweight} we see that each of the 
$p^\alpha$ assignments $\bh'\in C'(\B')$ has at most $p^{1-\alpha}/\N$
extensions $\bh$ to $\bx$ such that $\bh$ is overweight.
Therefore each $\B'<\B$ can contribute at most $p/\N$ elements of $C'(\B)$.
Since there are at most $\N$ choices of $\B'$, the lemma follows.
\end{proof}

%% file: proof-lower-bound.tex
\section{Proof of The Lower Bound}
\label{sec:lower:uniform}

In this section, we present the detailed proof of~\autoref{th:lower:uniform}.
Let $\E[|q(I)|]$ denote the expected number of answers to the query $q$.  

\begin{lemma} 
\label{lem:expected_size} 
The expected number of answers to $q$ is $\E[|q(I)|] = n^{k-a} \prod_{j=1}^{\ell} m_j$.
\end{lemma}

\begin{proof}
  For any relation $S_j$, and any tuple $\mathbf{a}_j \in [n]^{a_j}$, 
  the probability that $S_j$ contains $\mathbf{a}_j$
  is $P(\mathbf{a_j} \in S_j) = m_j / n^{a_j}$.  Given a tuple $\mathbf{a}
  \in [n]^k$ of the same arity as the query answer, let
  $\mathbf{a}_j$ denote its projection on the variables in $S_j$.  Then:
  \begin{align*}
    \E[|q(I)|]  
    &  = \sum_{\ba \in [n]^k}  P(\bigwedge_{j=1}^{\ell} (\mathbf{a}_j \in S_j)) 
     = \sum_{\ba \in [n]^k} \prod_{j=1}^{\ell} P(\mathbf{a}_j \in S_j)  \\
    & = \sum_{\ba \in [n]^k} \prod_{j=1}^{\ell} m_j n^{-a_j} 
    = n^{k-a} \prod_{j=1}^{\ell} m_j
\end{align*}
\end{proof}

Let us fix some server and let $m(I)$ denote the function 
specifying the message the server receives on input $I$.
For any fixed value $m$ of $m(I)$, we define the set of tuples of a relation $R$
of arity $a$ {\em known} by the server given message $m$ as
\begin{align*}
K_m(R)=\{t\in [n]^a\mid \mbox{ for all  }I,
m(I)=m\Rightarrow t\in R(I) \}
\end{align*}

We will particularly apply this definition with $R=S_j$ and $R=q$.
Clearly, an output tuple $\ba \in K_m(q)$ iff for every $j$,
$\ba_j \in K_m(S_j)$, where $\ba_j$ denotes the projection of $\ba$ on
the variables in the atom $S_j$.

\paragraph{Bounding the Knowledge of Each Relation}

Let us fix a server, and an input relation $S_j$. Observe that, for a randomly
chosen database $I$, $S_j$ is a randomly chosen subset of $[n]^{a_j}$ of size $m_j$.
Since there are $\binom{n^{a_j}}{m_j}$ such subsets, the number of bits necessary to represent
the relation is $\log \binom{n^{a_j}}{m_j}$. 

Let $m(S_j)$ be the part of the message $m$ received from the server that corresponds to 
$S_j$. The following lemma provides a bound on the expected knowledge $K_{m(S_j)}(S_j)$ 
the server may obtain from $S_j$:

\begin{lemma} \label{lem:entropy_ratio} Suppose that for all
  $\binom{n^{a_j}}{m_j}$ relations $S_j$, where $m_j \leq n^{a_j}/2$,
    $m(S_j)$ is at most $f_j \cdot \log \binom{n^{a_j}}{m_j}$ bits long. Then
  $\E[|K_{m(S_j)}(S_j)|] \le (\log e+1)  f_j \cdot m_j$.
\end{lemma}

  \begin{proof} 
  Let $m$ be a possible value for $m(S_j)$. Since $m$ fixes precisely $|K_m(S_j)|$ tuples of $S_j$,
  \begin{align}
  \log|\setof{S_j}{m(S_j)=m}| & \le \log \binom{n^{a_j}-|K_m(S_j)|}{m_j - |K_m(S_j)|} \nonumber
  \end{align}
We will apply the following lemma to the above equation.
\begin{lemma} \label{lem:binom_inequality}
For any $k \leq m \leq N/2$, and constant $c = \log e +1$,
\begin{align*}
\log \binom{N-k}{m-k} \leq \left(1-\frac{k}{cm} \right) \log \binom{N}{m}
\end{align*} 
\end{lemma}
  
 \begin{proof}
We have:
$$\log \binom{N-k}{m-k} =  \log \binom{N}{m} - \log \frac{[N]_k}{[m]_k}$$
where $[N]_k = N \cdot (N-1) \cdots (N-k+1)$. Since we also have that
$\binom{N}{m} = [N]_m / [m]_m$, it suffices to show that
$ \log ([N]_k / [m]_k) \geq k/(cm) \log \binom{N}{m}$, or equivalently
$  \frac{[N]_k}{[m]_k}  \geq \binom{N}{m}^{k/(cm)}$. To show this, notice 
that the following two inequalities hold:
\begin{align*}
 \frac{[N]_k}{[m]_k} \geq (N/m)^k,  \quad \binom{N}{m} \leq (eN/m)^m
\end{align*}
We are left to prove $(N/m)^k \geq (eN/m)^{k/c}$, or equivalently
$N/m \geq (eN/m)^{1/c}$, equivalently $(N/m)^{c-1} \geq e$. However, since
$m \leq N/2$, we have $N/m \geq 2$ and consequently $(N/m)^{c-1} \geq 2^{c-1}=e$,
which proves the lemma.
\end{proof} 
  
 We now have, for $c = \log e +1$: 
  \begin{align}
  \log|\setof{S_j}{m(S_j)=m}|  \le
  \left(1- \frac{|K_m(S_j)|}{cm_j} \right) \log \binom{n^{a_j}}{m_j} \label{eq:ratiobound}
  \end{align}
  
  We can bound the value we want by considering the binary entropy of the
  distribution $S_j$. By applying the chain rule for entropy, we have
  \begin{align}
  & H(S_j) \nonumber\\ 
  &=H(m(S_j))+ \sum_{m}P(m(S_j)=m)\cdot H(S_j|m(S_j)=m)\nonumber\\
  &\le f_j \cdot  H(S_j) +  \sum_{m}P(m(S_j)=m)\cdot H(S_j|m(S_j)=m)\nonumber\\
  &\le f_j \cdot  H(S_j)+
      \sum_{m}P(m(S_j)=m)\cdot (1-\frac{|K_m(S_j)|}{c \cdot m_j}) H(S_j)\nonumber\\
  &= f_j \cdot  H(S_j)+ 
      (1-\sum_{m}P(m(S_j)=m)\ \frac{ |K_m(S_j)|}{c \cdot m_j}) H(S_j)\nonumber\\
  &= f_j \cdot H(S_j)+  (1-\frac{ \E [|K_{m(S_j)}(S_j)|]}{c \cdot m_j}) H(S_j) \label{eq:entropy}
  \end{align}
  where the first inequality follows from the assumed upper bound on $|m(S_j)|$,
  the second inequality follows by \eqref{eq:ratiobound}, and the last two
  lines follow by definition.
  Dividing both sides of \eqref{eq:entropy} by $H(S_j)$ since $H(S_j)$ is not
  zero and rearranging we obtain the required statement.
  \end{proof}

\paragraph{Bounding the Knowledge of the Query}

We know that $\log \binom{n^{a_j}}{m_j}$ is the number of bits needed to
represent relation $S_j$. Since we want to show a lower bound of $L$, we 
assume that each server receives at most $L$ bits. 
Let us also fix some server. The message
$m=m(I)$ received by the server is the concatenation of $\ell$ messages, 
one for each input relation. $K_m(S_j)$ depends only on $m(S_j)$, so we 
can assume w.l.o.g. that $K_m(S_j) = K_{m(S_j)}(S_j)$.

The message $m(I)$  will contain at most $L$ bits. Now, for
each relation $S_j$, let us define 
$$f_j =  \frac{\max_{S_j} |m(S_j)|}{\log \binom{n^{a_j}}{m_j}} .$$
Thus, $f_j$ is the largest fraction of bits of $S_j$ that the server
receives, over all choices of the matching $S_j$.  We now derive an
upper bound on the $f_j$'s.  As we discussed before, each part
$m(S_j)$ of the message is constructed independently of the other
relations. Hence, it must be that $\sum_{j=1}^{\ell} \max_{S_j} |m(S_j)| \leq L$.
By substituting the definition of $f_j$ in this equation, we obtain that 
 $$L \geq \sum_{j=1}^{\ell} f_j \log \binom{n^{a_j}}{m_j}.$$
Using our assumptions on the choice of $n$, we can write:  
$$\log \binom{n^{a_j}}{m_j} \geq m_j \log (n^{a_j} /m_j) \geq m_j (a_j-\delta) \log n$$
Denoting $M_j = a_j m_j \log n$, we now have
$$L \geq \sum_{j=1}^{\ell} \frac{a_j-\delta}{a_j} f_j M_j 
\geq \min_j \frac{a_j-\delta}{a_j} \sum_j f_j M_j$$
where the quantity $C_0 = \min_j \frac{a_j-\delta}{a_j} $ is a positive constant (since $\delta \leq a_j$). 
  
For $\ba_j \in [n]^{a_j}$, let
$w_j({\ba_j})$ denote the probability that the server knows the tuple
$\ba_j$.  In other words $w_j(\ba_j) = P(\ba_j \in K_{m_j(S_j)}(S_j))$, where
the probability is over the random choices of $S_j$.

\begin{lemma} \label{lem:bounds}
For any relation $S_j$:
\begin{itemize}
\item[(a)] $\forall \ba_j \in [n]^{a_j}: w_j(\ba_j) \leq m_j/n^{a_j}$, and
\item[(b)] $\sum_{\ba_j \in [n]^{a_j}} w_j({\ba_j}) \leq 3 f_j \cdot m_j$.
\end{itemize}
\end{lemma}


Since the server receives a separate message for each relation
$S_j$, from a distinct input server, the events $\ba_1 \in K_{m_1}(S_1),
\ldots, \ba_\ell \in K_{m_\ell}(S_\ell)$ are independent, hence:
\begin{align*}
  \E[|K_{m(I)}(q)|] = \sum_{\ba \in [n]^k} P(\ba \in K_{m(I)}(q)) = \sum_{\ba \in [n]^k} \prod_{j=1}^{\ell}
  w_j({\ba_j})
\end{align*}
We now use Friedgut's inequality. Recall that in order to apply the inequality, we need to 
find a fractional edge cover. Let us pick any fractional edge packing
$\mathbf{u} = (u_1, \ldots, u_\ell)$. Given $q$,  consider the {\em extended query}, 
which has a new unary atom for each variable $x_i$:
\begin{align*}
  q'(x_1,\ldots, x_k) = S_1(\bar x_1), \ldots, S_\ell(\bar x_\ell),
  T_1(x_1), \ldots, T_k(x_k)
\end{align*}
For each new symbol $T_i$, define 
$  u_i' = 1 - \sum_{j: x_i \in \text{vars}(S_j)} u_j$.
Since $\mathbf{u}$ is a packing,  $u_i' \geq 0$. 
Let us define $\mathbf{u}'=(u_1', \ldots, u_k')$.

\begin{lemma} 
\label{lem:tight_cover}
(a) The assignment $(\mathbf{u},\mathbf{u}')$ is both a tight
fractional edge packing and a tight fractional edge cover for $q'$. 
(b) $\sum_{j=1}^\ell a_j u_j + \sum_{i=1}^k u'_i = k$
\end{lemma}

\begin{proof}
  (a) is straightforward, since for every variable $x_i$ we have $u_i'+\sum_{j:
    x_i \in \text{vars}(S_j)} u_j = 1$.  Summing up:
  \begin{align*}
k=\sum_{i=1}^k \left( u_i'+\sum_{j: x_i \in \text{vars}(S_j)} u_j \right) = 
\sum_{i=1}^k u_i' + \sum_{j=1}^\ell a_j u_j 
  \end{align*}
which proves (b).  
\end{proof}

We will apply Friedgut's inequality to the extended query.  
Set the variables $w(-)$ used in Friedgut's
inequality as follows:
\begin{align*}
  w_j(\ba_j) = & P(\ba_j \in K_{m_j(S_j)}(S_j)) \mbox{ for $S_j$,
    tuple $\ba_j \in [n]^{a_j}$} \\
  w_i'(\alpha) = & 1\kern 1.1in  \mbox{ for $T_i$, value $\alpha \in [n]$}
\end{align*}

Recall that, for a tuple $\ba \in [n]^k$ we use $\ba_j \in [n]^{a_j}$
for its projection on the variables in $S_j$; with some abuse, we
write $\ba_i \in [n]$ for the projection on the variable $x_i$.
Assume first that $u_j>0$, for $j=1,\ell$.  Then:
\begin{align*}
\allowdisplaybreaks
\E[|K_m(q)|] 
& =  \sum_{\ba \in [n]^k} \prod_{j=1}^{\ell} w_j({\ba_j}) \\
& = \sum_{\ba \in [n]^k} \prod_{j=1}^{\ell} w_j({\ba_j})\prod_{i=1}^k w_i'({\ba_i})\\
&\leq \prod_{j=1}^{\ell} \left( \sum_{\ba \in [n]^{a_j}} w_j({\ba})^{1/u_j} \right)^{u_j}
\prod_{i=1}^{k} \left( \sum_{\alpha \in [n]} w'_i(\alpha)^{1/u_i'} \right)^{u_i'} \\
 &= \prod_{j=1}^{\ell} \left( \sum_{\ba \in [n]^{a_j}} w_j({\ba})^{1/u_j} \right)^{u_j} \prod_{i=1}^k n^{u_i'}
\end{align*}
Note that, since $w'_i(\alpha) = 1$ we have $w'_i(\alpha)^{1/u_i'} =
1$ even if $u_i'=0$.  Write $w_j({\ba})^{1/u_j} = w_j({\ba})^{1/u_j-1}
w_j({\ba})$, and use~\autoref{lem:bounds} to obtain:
\begin{align*}
\sum_{\ba \in [n]^{a_j}} w_j({\ba})^{1/u_j} 
& \leq (m_j / n^{a_j})^{1/u_j-1} \sum_{\ba \in [n]^{a_j}} w_j({\ba}) \\
& \leq (m_j  n^{-a_j})^{1/u_j-1}  3 f_j \cdot m_j \\
&  = 3 f_j \cdot m_j^{1/u_j} \cdot n^{( a_j -a_j/u_j)}
\end{align*}
Plugging this in the bound, we have shown that:
\begin{align}
\E[|K_m(q)|]& \leq \prod_{j=1}^{\ell} (3 f_j \cdot m_j^{1/u_j} \cdot n^{( a_j -a_j/u_j)})^{u_j} \cdot \prod_{i=1}^k n^{u_i'} \nonumber\\
  &= \prod_{j=1}^{\ell } (3f_j)^{u_j} \cdot \prod_{j=1}^{\ell} m_j \cdot n^{(\sum_{j=1}^{\ell}a_j u_j - a)} \cdot n^{\sum_{i=1}^k u_i'} \nonumber\\
  &= \prod_{j=1}^{\ell } (3f_j)^{u_j} \cdot \prod_{j=1}^{\ell} m_j \cdot  n^{-a +(\sum_{j=1}^{\ell}a_j u_j+ \sum_{i=1}^k u_i')} \nonumber\\
  &= \prod_{j=1}^{\ell } (3f_j)^{u_j} \cdot \prod_{j=1}^{\ell} m_j \cdot  n^{k-a} \nonumber\\
   & = \prod_{j=1}^{\ell } (3f_j)^{u_j} \cdot  \E[|q(I)|]  \label{eq:end}
\end{align}
If some $u_j=0$, then we can derive the same lower bound as follows:
We can replace each $u_j$ with $u_j+\delta$ for any $\delta>0$ still
yielding an edge cover. 
Then we have $\sum_j a_j u_j + \sum_i u_i' = k + a\delta$,
and hence an extra factor $n^{a\delta}$ multiplying the term
$n^{k-a}$ in \eqref{eq:end}; however, we
obtain the same upper bound since, in the limit as $\delta$ approaches 0,
this extra factor approaches 1.

Let $f_q = \prod_{j=1}^{\ell }(3f_j)^{u_j}$; the final step is to upper bound the quantity $f_q$
using the fact that $\sum_{j=1}^{\ell} f_j M_j \leq L/C_0$. Indeed:
\begin{align*}
 f_q  = \prod_{j=1}^{\ell} (3f_j)^{u_j} 
 & = \prod_{j=1}^{\ell} \left( \frac{f_j M_j}{u_j} \right)^{u_j}
     \prod_{j=1}^{\ell} \left( \frac{3 u_j}{M_j} \right)^{u_j} \\
& \leq \left( \frac{\sum_{j=1}^{\ell} f_j M_j}{\sum_j u_j} \right)^{\sum_j u_j}
\prod_{j=1}^{\ell} \left( \frac{3 u_j}{M_j} \right)^{u_j}  \\
& \leq \left( \frac{L/C_0}{\sum_j u_j} \right)^{\sum_j u_j}
\prod_{j=1}^{\ell} \left( \frac{3 u_j}{M_j} \right)^{u_j}  \\
& = \prod_{j=1}^{\ell} \left( \frac{L}{M_j} \right)^{u_j}
\prod_{j=1}^{\ell} \left( \frac{3 u_j}{C_0 \sum_j u_j} \right)^{u_j}  \\
\end{align*}

Here, the first inequality comes from the weighted version of the Arithmetic Mean-Geometric
Mean inequality. Notice also that, since $u_j \leq 1$, we have that  $(u_j / \sum_j u_j)^{u_j} \leq 1$, and thus:
\begin{align*}
\prod_{j=1}^{\ell} \left( \frac{3 u_j}{C_0 \sum_j u_j} \right)^{u_j}  \leq (3/C_0)^{\sum_j u_j} 
\end{align*}

Recall that $u = \sum_j u_j$, therefore, 
\begin{align*}
  \E[|K_m(q)|] \leq  (3/C_0)^u \prod_{j=1}^{\ell} \left( \frac{L}{M_j}
  \right)^{u_j} = (3/C_0)^u \frac{L^u}{K(\bu,\bM)}
\end{align*}
%
%
%
which completes the proof of~\autoref{th:lower:uniform}.  (Recall that
our $L$ denotes the load of an arbitrary server, which was denoted
$L_i$ in the statement of the theorem.)

%

%% file: hashing.tex
\section{Hashing}
\label{sec:hashing}

In this section, we discuss hashing and in particular prove \autoref{lemma:hashing}. 
Let $R(A_1, \dots, A_r)$ be a relation of arity $r$ of size $m$.
Further, let $[n]$ be the domain for each attribute. 
For any $i \in [n]$, we define its degree on attribute $A_j$ as
$d_j(i) = |\setof{t \in R}{t[A_j] = i}|$.
We extend this definition to sets of attributes for $S\subseteq [r]$ and
for $I=(i_j)_{j\in S}\in [n]^{S}$, 
we let $R_{S}(I)=\{t\in R\mid t[A_j] = i_j,\ \forall j\in [S]\}$ and
We write $d_{S}(I)=|R_S(I)|$.

In our setting, we wish to hash each tuple from $R$ on $p = p_1 \times \dots \times p_r$ bins: each bin corresponds to a point in the $r$-dimensional space $\mathcal{A} = [p_1] \times \dots [p_r]$. For each attribute $A_j$, we are equipped with a perfect hash function $h_j: [n] \rightarrow [p_j]$. Each tuple $t \in R$ will be sent to the bin $B(t)$ as follows:
$$B(t) = (h_1(t[A_1]), \dots, h_r(t[A_r]))$$
where $t[A_j]$ is the value of the tuple $t$ at attribute $A_j$. 
For a bin $b \in \mA$, let $b_j$ denote its coordinate at the $j$-th dimension. Moreover, let $L(b)$ denote its {\em load}, in other words:
$$ L(b) = | \setof{t \in R}{B(t)  = b} |$$

\begin{lemma}
For every bin $b \in \mathcal{A}$, $ \E[L(b)] = m/p$.
\end{lemma} 

\begin{proof}
First, note that for a given tuple $t \in R$,
$Pr[B(t) = b] = \prod_{j=1}^r Pr[t[A_j] = b_j] = \prod_{j=1}^r (1/p_j) = 1/p$.
Since we have $m$ tuples, using linearity of expectation we obtain the desired
result.
\end{proof}

Thus, the tuples are in expectation distributed equally among the bins.
However, this does not tell us anything about the maximum load of the bins,
$L = \max_{b} \{ L(b)\}$.
An easy observation about $L$ is that $m / p \leq L \leq m$.
Ideally, we want $L$ to be equal to the expected load $m/p$, but this is not
always possible, as the next example shows.

\begin{example}
Let $R = \setof{(1, \dots, 1, i)}{i \in [m]}$. In this case, every bin that
receives  a tuple from $R$ will have the same first $r-1$ coordinates,
and hence the $m$ tuples will be distributed among only $p_r$ servers.
This implies that $L \geq m/p_r$.
Using similar reasoning, there will be some instance of $R$ such that
$L \geq m / \min_j \{p_j\}$.
\end{example}

Let us set up some notation.
Let $X_{b, t}$ denote the indicator for the event that the tuple $t$ ends up
in bin $b$.
Then, $L(b) = \sum_{t \in R} X_{b,t}$. 
Additionally, for $j = 1, \dots, r$, and $b_j \in [p_j]$,
let $L^j(b_j) = \sum_{b': b'_j = b_j} L(b')$.
Finally, let  $L^j = \max_{b_j \in [p_j]} \{ L^j(b_j)\}$.

We start our analysis with the case where for every attribute $A_j$, and every
$i \in [n]$, the degree $d_j(i)$ is either 0 or 1, or equivalently, each value
appears at most once in every attribute.
In this case, $R$ is a {\em matching} relation.

\begin{lemma}
\label{lem:load_matching}
Assume that, for every attribute $A$ and every value $i \in [n]$,
$d_A(i) \leq 1$.
Then, 
$$Pr[L > 3m/p] \leq p e^{-m/p}.$$
\end{lemma}

\begin{proof}
Since the tuples contains disjoint values, each tuple is hashed independently
and uniformly at random to one of the $p$ bins.
Thus, we can apply directly \autoref{cor:binbins} to derive the result.
\end{proof}

\begin{lemma}
\label{lem:row_bound}
Suppose that for some attribute $A_j$, for every $i$, $d_{A_j}(i) \leq am /p_j$ for $a>1/\ln(1/\delta)$ and $0<\delta\le 1/p_j$.   
Then 
$$Pr[ L^j > 3\ln(1/\delta)am/ p_j] \leq p_j \delta.$$ 
\end{lemma}

\begin{proof}
First, notice that we have the following equivalent expression for $L^j(b_j)$, $b_j \in [p_j]$:
$$ L^j(b_j) = \sum_{i \in [n]}  d_{A_j}(i) \cdot X_{h_j(i) = b_j} $$
where $X_{h_j(i)=b_j}$ denotes the indicator variable for the event that the
hash function $h_j$ sends node $i$ to row a bin with $j$-th coordinate $b_j$.
Since our hash function is perfectly random, we have a weighted balls-into-bins
problem, where each ball is a node $i$ of weight $d_{A_j}(i)$.
The total weight of the balls is $m$ and $d_{A_j}(i) \leq am/ p_j $ for every
$i \in [n]$; since we have $p_j$ bins, we can apply
\autoref{lem:wballs_in_bins} to obtain the probability bound.
\end{proof}

\begin{lemma}
\label{lem:rectangle_hash}
Let $m \geq p^2$.
If for every set of attributes $S\subseteq [r]$ and vector of values
$I=(i_j)_{j\in S} \in [n]^{S}$, $d_S(I)\le am/\prod_{j\in S} p_j$
where $a\ge 1/\ln (1/\delta))$ and $0<\delta< 1$,
then
$$ Pr \left[L > [3\ln(1/\delta)a]^r \frac{m}{p} \right]
< 2^{r+1} p^2 \delta/(3\ln(1/\delta)a^2). $$
\end{lemma}

\begin{proof}
The proof uses a recursive argument. 
For each $b_1 \in [p_1]$,  we look at the tuples that have been hashed to all
the bins $b'$ such
that $b'_1 = b_1$.  These tuples form a sub-relation of $R$,
which we denote by $R^{b_1}$; 
for $S\subseteq [r]-\{1\}$, and $I=(i_j)_{j\in S}$, let 
$$d_{S}^{b_1}(I)=\sum_{i_1,\ h_1(i_1)=b_1} d_{S\cup\{1\}}(i_1,I)$$
denote the induced degrees of partial vectors of values in $R^{b_1}$.

We will show next that, with some high probability,
our problem reduces to hashing the relations $R^{b_1}$ over
$p_2 \times \dots \times p_r$ bins
such that for each such relation:
\begin{packed_item}
\item $|R^{b_1}| \leq 3\ln(1/\delta)am/p_1$ which we denote by $m_1$, and
\item for each subset of attributes $S\subseteq [r]-\{1\}$ and vector
of values $I=(i_j)_{j\in S}\in [n]^{S}$, we have
that $d_{S}^{b_1}(I) \leq am_1/\prod_{j\in S} p_j$.
\end{packed_item}

The first task it to get a bound on the size of $R^{b_1}$. Since 
$|R^{b_1}| = L^1(b_1)$, we can apply \autoref{lem:row_bound} to obtain that
$Pr[\max_{b_1} \{ |R^{b_1}| \} >  3\ln(1/\delta)am/p_1] \leq p_1 \delta$. 

The next task is to bound the degrees of vectors of vertices in each
subrelation $R^{b_1}$ so that we can apply a recursive argument.
In particular, for each $S\subseteq [r]-\{1\}$ and $I=(i_j)_{j\in S}$ we will
show that with high probability, $d_S^{b_1}(I) \le am_1/\prod_{j\in S} p_j$.

Fix some $S\subseteq [r]-\{1\}$.
For $I=(i_j)_{j\in S}\in [n]^S$, if
$d_S(I)\le am_1/\prod_{j\in S} p_j =3\ln(1/\delta) a^2 m/(p_1 \prod_{j\in S} p_j)$
then we automatically have
$d_S^{b_1}(I) \le am_1/\prod_{j\in S} p_j$ for every $b_1\in [p_1]$.
Since $|R|\le m$, there are fewer than
$p_1 \prod_{j\in S} p_j/(3a^2\ln(1/\delta))$
choices $I\in [n]^S$ such that
$$d_S(I)> am_1/\prod_{j\in S} p_j=3\ln(1/\delta) a^2 m/(p_1 \prod_{j\in S} p_j).$$
We are interested in bounding the load from each of these choices on a given
$R_{b_1}$.
By our hypothesis on $R$, for each such $I$ there are
$d_{S\cup \{1\}}(i_1,I)\le am/(p_1\prod_{j\in S} p_j)$ elements of 
$R$ that have any fixed set of values $(i_1,I)=(i_j)_{j\in S\cup \{1\}}$.
For each of the choices of $I$ the load is given by the load of 
a balls-in-bins problem with $p_1$ bins,
total weight $d_S(I)\le am/\prod_{j\in S} p_j$, and maximum weight per ball
of $\max_{i_1} d_{S\cup \{1\}}(i_1,I)\le am/(p_1\prod_{j\in S} p_j)$.
By \autoref{lem:wballs_in_bins} the maximum value of
$d_S^{b_1}(I)$ for each $b_1$ is at most
$3\ln(1/\delta)am/(p_1 \prod_{j\in S} p_j)=am_1/\prod_{j\in S} p_j$ except
with probability at most $p_1 \delta$.
For this value of $S$, there are at most
$p_1 \prod_{j\in S} p_j/(3a^2\ln(1/\delta))$ choices of $i_j$ to which
we need to apply this argument.  
Running over all these choices of $S$, we
get a total failure probability for the reduction step of at most
$$[p_1+p_1^2\sum_{S\subseteq [r]-\{1\}} \prod_{j\in S} p_j/(3a^2\ln(1/\delta))]\delta,$$
which we denote by $P(m,p_1,\ldots,p_r)$.

We now have the precisely analogous conditions for each of the $p_1$ choices of
$R^{b_1}$ as we had initially for $R$,
except with $m_1=3\ln(1/\delta)am/p_1$ replacing $m$ and the first
coordinate removed.  
Let $K=3\ln(1/\delta)a$.
Repeating this $j\le r$ times, for each $j$ if all the reduction steps succeed
we have total load at most $m_j=K^j m/\prod_{i\le j} p_i$ in the
bins consistent with $(b_1,b_2,\ldots, b_j)$ and the total failure probability
for all $r$ levels is at most
\begin{align*}
&\delta\cdot \sum_{j=1}^r P((K^j m/\prod_{i<j} p_i),p_j,\ldots,p_r) \prod_{i<j} p_i \\
&\le \delta \cdot \left [\sum_{j=1}^r (\prod_{i=1}^j p_i) [1+\sum_{S\subseteq [r],\ \min(S)=j} (\prod_{i\in S} p_i)/(aK)]\right]\\
&\le 2\delta\cdot \sum_{j=1}^r [\sum_{S\subseteq [r],\ \min(S)=j} p_j p/(aK)]\\
&< 2^{r+1} \delta  p^2/(aK).
\end{align*}
which is what we wanted to prove.
\end{proof}

As a corollary, we have:

\begin{corollary}
If for every subset of attributes $S\subseteq [r]$ and vector of values
$I=(i_j)_{j\in S}$ we have that
$d_S(I) \leq am/\prod_{j\in S} p_j$, then the probability that the load
$L > [9a\ln p]^r m/p$ is polynomially small in $p$.
\end{corollary}

\begin{proof}
Let $\delta = p^{-3}$ so $\ln(1/\delta)=3\ln p$ and apply
\autoref{lem:rectangle_hash}.
\end{proof}

\begin{proposition}
For any relation $R$,  we have:
$$ Pr \left[ L > \frac{(3r+1)m}{\min_j(p_j)} \right]
 \leq 2 \sum_j p_j e^{-\left(m/\min_j(p_j)\right)^{1/r}/3}.$$
\end{proposition}

\begin{proof}
For attribute $A_j$, let us fix a threshold $\tau_j = m / (p_j \ln (1/\delta))$, where $\delta$ is a 
parameter we will choose later. Let $R^j$ denote the subset of $R$ that contains tuples $t$ such that
$d_{j}(t[A_j]) \leq \tau_j$. 

If $L^j_0$ denotes the load $L^j$ for the instance $R^j$, assuming the worst-case scenario where $|R_j|=m$, we can apply directly \autoref{lem:row_bound} to obtain that $Pr[L^j_0 > 3m/p_j] \leq p_j \delta$. Notice also that the maximum load of any bin for $R^j$ is upper bounded by $L^j_0$. 

Next consider the tuples in $R' = R \setminus \bigcup_j R^j$, which are the tuples $t$ such that for
every attribute $A_j$, we have $d_{j}(t[A_j]) > \tau_j$. The main observation
here is that the tuples in $R'$ are very few. Indeed for each $A_j$ there
are at most $m/\tau_j=p_j\ln(1/\delta)$ values of the $A_j$ attribute that
appear in tuples of $R'$ and hence 
$$ |R'| \leq \prod_{i=1}^r (p_j\ln(1/\delta))= p \ln^r(1/\delta).$$
Since only $p_j \ln (1/\delta)$ values of the $A_j$ attribute appear in $R'$
and there are $p_j$ values of the bin coordinate for $A_j$,
\autoref{cor:binbins} implies that the probability that some bin will receive
$> 3\ln(1/\delta)$ values from the attribute $A_j$ is at most $p_j \delta$.
Hence, by applying a union bound, the probability that there exists a bin
that receives $> 3\ln(1/\delta)$ values on any coordinate is at most
$\delta \sum_j p_j$.
Notice that if this bad event does not happen, each bin will hold at most
$[3\ln(1/\delta)]^r$ tuples from $R'$.

Since $R = R' \cup (\bigcup_j R_j )$, we can add up the probabilities of the bad events to conclude that
$$\Pr \left[ L > \frac{3rm}{\min_j(p_j )} + [3\ln(1/ \delta)]^r \right]
 \leq 2\delta \sum_j p_j$$
The claim is obtained by choosing the value of $\delta$ 
such that $[3\ln (1/\delta)]^r = m / \min_j (p_j)$; 
that is $\delta = \exp (-[m/ \min_j(p_j)]^{1/r}/3)$.
\end{proof}

We can easily now obtain the following corollary.

\begin{corollary}
If $m \geq p \ln^c(p)$ for some large constant $c$, then $L > \frac{(3r+1)m}{\min_j p_j}$ with probability polynomially small in $p$.
\end{corollary}

\section{Balls Into Bins}

Here, we study the standard setting where balls (possibly weighted) are thrown independently and uniformly at random into bins, and we are interested in the maximum weight that each bin will receive.

\begin{lemma}
\label{lem:wballs_in_bins}
Let $0 < \delta < 1$.  
Suppose that we throw weighted balls independently and uniformly at random
into $p$ bins, such that:
\begin{packed_enum}
\item The total weight of the balls is at most $m$.
\item The weight of every ball at most $B = a m / p$, for 
$a\ge 1/\ln (1/\delta)$.
\end{packed_enum}
Then, the probability that the maximum weight per bin is
$> 3\ln (1/\delta)am/p$ is at most $p \delta$.
\end{lemma}

\begin{proof}
By convexity, the maximum load is maximized for the weight
distribution that
has $n=m/B$ balls of weight $B$; hence, we will upper bound the probability of
this event.
Consider some bin $j=1, \dots, p$.
Each ball falls into this bin independently with probability $1/p$.
If $X_{ij}$ denotes the event that the ball $i=1, \dots, n$ falls into the bin
$j$, the random variable $W_j = B \sum_i X_{ij}$ denotes the weight of the
bin $j$. We now apply a form of Chernoff bound to the random variable
$X_j = \sum_i X_{ij}$ which has $\mu = E[X_j] = m/(p B)=1/a$ which says that
$$ Pr[X_j > (1+\epsilon) \mu] \leq 
\begin{cases}\exp(-\mu \epsilon^2/4)\qquad \epsilon<2e-1\\
2^{-\epsilon \mu}\qquad \epsilon\ge 2e-1.
\end{cases}
$$ 
We set $\epsilon=2a\ln(1/\delta)$ and observe that
$\epsilon\le 3a\ln(1/\delta)-1\ge 2$ since $a\ge 1/\ln(1/\delta)$.   
We have two cases.   For $2\le \epsilon< 2e-1$ we use the first
bound to obtain $\Pr[X_j > 3\ln(1/\delta)] \le
\exp(-2\ln(1/\delta)\epsilon^2/4)\le \exp(-2\ln(1/\delta)) = \delta^2 < \delta$; 
for $2a\ln(1/\delta)\ge 2e-1$ we use the second bound to obtain
$\Pr[X_j > 3\ln(1/\delta)] \le 2^{-\epsilon\mu}\le 2^{-2\ln(1/\delta)}<\delta$. 
Applying a union bound we have
$\Pr[\max_j W_j> 3\ln(1/\delta)am/p]\le p\delta$ as required.
\end{proof}

\begin{corollary}
\label{cor:binbins}
Suppose that we throw $m$ balls independently and uniformly at random into $p$
bins. Then, the probability that the maximum number of balls per bin is
$> 3m/p$ is at most $p e^{-m/p}$.
\end{corollary}

\begin{proof}
We apply \autoref{lem:wballs_in_bins} with $\delta=e^{-m/p}$,
$a=1/\ln(1/delta)=p/m$, and hence $B=1$.
\end{proof}

Notice that, in order to get any meaningful probability from
\autoref{cor:binbins}, it must be that $m \geq p \ln(p)$.